%% file: main.tex
\pdfoutput=1
\PassOptionsToPackage{dvipsnames}{xcolor}
\PassOptionsToPackage{bookmarks=false}{hyperref}
\documentclass[dvipsnames,format=sigconf]{acmart}
\usepackage[english]{babel}
\usepackage[utf8]{inputenc}
\usepackage{csquotes} 
\usepackage{xparse}

\usepackage{csquotes}

\usepackage{amsmath}
\usepackage{mathtools}
\usepackage[capitalise,nameinlink]{cleveref}
\usepackage{algorithm}
\usepackage{algpseudocode}
\crefname{section}{Sect.}{Sect.}
\Crefname{section}{Section}{Sections}
\usepackage{xspace}
\usepackage{graphicx}
\usepackage{caption}
\usepackage{subcaption}
\usepackage{centernot}
\usepackage{color}
\usepackage{thm-restate}

\usepackage{listings}
\usepackage[color=gray!20, linecolor=gray!80!black]{todonotes}
\usepackage[binary-units]{siunitx}

\usepackage{listings}
\usepackage{multirow}
\usepackage{ifthen}
\usepackage[inline,shortlabels]{enumitem}
\usepackage{lstautogobble}  

\usepackage{soul}
\usepackage{wrapfig}
\usepackage{tikz}
\usetikzlibrary{arrows,automata,positioning,shapes,backgrounds,patterns,calc}

\usepackage{pgfplots}

\include{commands}
\include{style}

\newboolean{intronecessitymerges}
\setboolean{intronecessitymerges}{false}
\newboolean{hideproofs}
\setboolean{hideproofs}{false}

\usepackage{dsfont}
\usepackage{upgreek}
\usepackage{textcomp}

\begin{document}


\title{Conservative Hybrid Automata from Development Artifacts}




\author{Niklas Metzger}
\affiliation{%
	\institution{CISPA Helmholtz Center for Information Security}
	\city{Saarbrücken}
	\country{Germany}
}
\email{niklas.metzger@cispa.de}
\author{Sanny Schmitt}
\affiliation{%
	\institution{CISPA Helmholtz Center for Information Security}
	\city{Saarbrücken}
	\country{Germany}
}
\email{sanny.schmitt@stud.uni-saarland.de}
\author{Maximilian Schwenger}
\affiliation{%
	\institution{CISPA Helmholtz Center for Information Security}
	\city{Saarbrücken}
	\country{Germany}
}
\email{maximilian.schwenger@cispa.de}


\acmConference[HSCC'22]{Hybrid Systems: Computation and Control}{May 4--6, 2022}{Milan, Italy}


\begin{abstract}
The verification of cyber-physical systems operating in a safety-critical environment requires formal system models.
The validity of the verification hinges on the precision of the model: possible behavior not captured in the model can result in formally verified, but unsafe systems.
Yet, manual construction is delicate and error-prone while automatic construction does not scale for large and complex systems.
As a remedy, this paper devises an automatic construction algorithm that utilizes information contained in artifacts of the development process:  a runtime monitoring specification and recorded test traces.
These artifacts incur no additional cost and provide sufficient information so that the construction process scales well for large systems.
The algorithm uses a hybrid approach between a top-down and a bottom-up construction which allows for proving the result \emph{conservative}, while limiting the level of over-ap\-prox\-i\-ma\-tion.
\end{abstract}

\maketitle

\input{source/introduction}
\input{source/problem}

\input{source/preliminaries}
\input{source/construction}

\input{source/completeness}
\input{source/evaluation}

\input{source/related_work}
\input{source/conclusion}

\bibliography{bibliography}
\newpage
\input{source/appendix}
\end{document}

%% file: commands.tex
\DeclareFontFamily{U}{MnSymbolC}{}
\DeclareSymbolFont{MnSyC}{U}{MnSymbolC}{m}{n}
\DeclareFontShape{U}{MnSymbolC}{m}{n}{
    <-6>  MnSymbolC5
   <6-7>  MnSymbolC6
   <7-8>  MnSymbolC7
   <8-9>  MnSymbolC8
   <9-10> MnSymbolC9
  <10-12> MnSymbolC10
  <12->   MnSymbolC12%
}{}
\DeclareMathSymbol{\powersetsym}{\mathord}{MnSyC}{180}

\newcommand*{\eg}{e.g.\@\xspace}
\newcommand*{\ie}{i.e.\@\xspace}

\makeatletter
\newcommand*{\etc}{%
  \@ifnextchar{.}%
    {etc}%
    {etc.\@\xspace}%
}
\newcommand*{\cf}{%
  \@ifnextchar{.}%
    {cf}%
    {cf.\@\xspace}%
}
\newcommand*{\etal}{%
  \@ifnextchar{.}%
    {et~al}%
    {et~al.\@\xspace}%
}
\makeatother

\NewDocumentCommand{\twopartdef}{ m m m o }%
{
  \left\{
    \begin{array}{ll}
      #1\quad & \mbox{if } #2 \\
      #3\quad & \IfValueTF{#4}{\mbox{if } #4}{\mbox{otherwise}}
    \end{array}
  \right.
}
\NewDocumentCommand{\threepartdef}{ m m m m m o }%
{
  \left\{
    \begin{array}{ll}
      #1\quad & \mbox{if } #2 \\
      #3\quad & \mbox{if } #4 \\
      #5\quad & \IfValueTF{#6}{\mbox{if } #6}{\mbox{otherwise}}
    \end{array}
  \right.
}

\NewDocumentCommand{\dist}{ o m m }%
{
\mathit{dist}(#1,#2)%
}

\newcommand*{\real}{\mathds{R}}
\newcommand*{\rplus}{\real_{\geq 0}}
\newcommand*{\rn}{\real^n}
\renewcommand*{\natural}{\mathds{N}}

\newcommand*{\interval}{\mathcal{I}}
\renewcommand*{\emptyset}{\varnothing}

\newcommand*\tracecompliant{\triangleright}

\renewcommand*{\epsilon}{\varepsilon}
\renewcommand*{\phi}{\varphi}
\newcommand*{\from}{\colon}

\newcommand*\convex[2]{%
  \mathit{conv}\ifthenelse{\equal{#1}{}}{}{({#1},{#2})}%
}
\newcommand*\Convex[1]{%
  \mathit{Conv}\ifthenelse{\equal{#1}{}}{}{({#1})}%
}
\DeclareMathOperator{\indeg}{indeg}
\DeclareMathOperator{\outdeg}{outdeg}

\newcommand*\bigO[1]{\mathcal{O}(#1)}

\newcommand*\manmerge{{\merge^{\kern-.2em\ast}}}
\newcommand*\given{\mid}
\newcommand*\guardspec{{\Gamma^\spec}}

\newcommand*\godtrace{{\widetilde\trace}}
\newcommand*\godtraces{{\widetilde\traces}}

\newcommand*\chull[1]{\mathit{chull}({#1})}
\DeclareMathOperator{\solve}{solve}
\NewDocumentCommand\projection{m O{\godtraces}}{\left.{#1}\right\vert_{{#2}}}
\newcommand*\dimaccess[2]{\left({#1}\right)_{{#2}}}

\newcommand*{\flow}{\ensuremath{\mathit{flow}}\xspace}

\newcommand*\commitpartition[2]{{{#1}\kern-.15em\downarrow_{#2}}}
\NewDocumentCommand\representative{O{} m}{ \left\llbracket {#2} \right\rrbracket_{#1} }

\newcommand*\rtlola{\textsc{RTLola}\xspace}

\newcommand*\aut{\ensuremath{\mathcal{H}}\xspace}
\NewDocumentCommand{\constraut}{ s D<>{\traces} D<>{\spec} }{ 
  \IfBooleanTF{#1}
  { \aut^+_{{#2},{#3}} }
  { \aut^+ }
}
\newcommand*\finaut{\ensuremath{\mathcal{A}}\xspace}

\newcommand*\trace{\ensuremath{\pi}\xspace}
\newcommand*\traces{\ensuremath{\Pi}\xspace}

\newcommand*\spec{\ensuremath{\Phi}\xspace}
\newcommand*\mode{\ensuremath{\mu}\xspace}
\newcommand*\sstate{\ensuremath{s}\xspace}
\newcommand*\modes{\ensuremath{M}\xspace}

\newcommand*\edges{\ensuremath{E}\xspace}
\newcommand*\edge{\ensuremath{e}\xspace}
\newcommand*\guard{\ensuremath{\gamma}\xspace}

\newcommand*\actions{\ensuremath{\Lambda}\xspace}
\newcommand*\action{\ensuremath{\lambda}\xspace}
\newcommand*\partition{\ensuremath{\mathcal{P}}\xspace}
\newcommand*\eqclass{\ensuremath{\zeta}\xspace}
\newcommand*\initialstate{\ensuremath{\initial{s}}\xspace}
\newcommand*\initial[1]{#1_I}

\newcommand*\delay{\ensuremath{\delta}\xspace}
\newcommand*\neutralelem{\mathds{1}}
\newcommand*\modemap{\ensuremath{\psi}\xspace}

\newcommand*\bisim{\approx}
\newcommand*\discretebisim{\mathrel{\bisim}}  
\newcommand*\actionsim{\mathrel{\sim_{\exists\action}}}
\newcommand*\termsim{\mathrel{\sim_{\bot}}}
\newcommand*\mergesim{\mathrel{\sim_{M}}}
\newcommand*\specrel{\mathrel{\sim_{\spec}}}
\newcommand*\reconrel{\mathrel{\sim_{+}}}

\newcommand*\refinerel{\mathrel{\sim_{\alpha}}}

\NewDocumentCommand{\Set}{ s m }{
  \IfBooleanTF{#1}%
  {\{#2\}}%
  {\left\{{#2}\right\}}%
}
\NewDocumentCommand{\card}{ s m }{
  \IfBooleanTF{#1}%
  {|#2|}%
  {\left\lvert{#2}\right\rvert}%
}

\newcommand*\lang[1]{\mathcal{L}(#1)}

\newcommand*{\modename}[1]{\textsc{#1}}

\newcommand\donotshow[1]{}

\definecolor{bluekeywords}{rgb}{0.13, 0.13, 1}
\definecolor{greentypes}{rgb}{0, 0.5, 0}
\definecolor{redstrings}{RGB}{171, 114, 2}
\definecolor{graynumbers}{rgb}{0.5, 0.5, 0.5}
\definecolor{goldcomments}{rgb}{0.6, 0.4, 0.08}
\definecolor{monitorblue}{RGB}{18, 163, 38}

\lstdefinelanguage{Lola}{
  keywords=[0]{input, output, trigger},
  keywordstyle=[0]\bfseries\color{bluekeywords},
  keywords=[1]{if, then, else, aggregate, defaults, offset, last},
  keywords=[2]{Int8, Int16, Int32, Int64, UInt8, UInt16, UInt32, UInt64, Bool, Float32, Float64, @1Hz, @5Hz, @10Hz, @100mHz, @1kHz},
  keywordstyle=[2]\color{greentypes},
    sensitive=false,
    comment=[l]{//},
    morecomment=[s]{/*}{*/},
    morestring=[b]',
    morestring=[b]"
}

\definecolor{prelude}{HTML}{F79A3D}
\definecolor{loop}{HTML}{2874ae}
\definecolor{prefix}{HTML}{61B940}
\definecolor{postfix}{HTML}{A370AB}

\newlength{\arrow}
\newcommand*{\myrightarrow}[1]{\xrightarrow{\mathmakebox[\arrow]{#1}}}

%% file: style.tex
\definecolor{CommentColor}{rgb}{0.16,0.60,0.16}

\tikzstyle{state}=[
	draw, 
	rectangle, 
	fill=none, 
	align=left, 
	minimum height = 3.0cm, 
	rounded corners=1mm, 
	minimum width = 3cm
]

%% file: source/introduction.tex
\section{Introduction} 

Hybrid systems connect the discrete, logical world of computers with the continuous, unpredictable real world.
In the last decades, they became an inevitable part of society by controlling essential infrastructures such as power plants, commercial airplanes, and cars.
A mathematical model for such systems are hybrid automata, which combine information regarding their discrete control structure and continuous physical behavior.
They can be used in all phases of development.  
Before deployment, they allow for analysis and verification of critical properties.  
During deployment, time-bounded prediction and identification of anomalous or unexpected behavior based on the model is possible. 
Lastly, recorded mission data enables a post-mortem analysis as well.
Albeit indisputably beneficial for the development process, designing a hybrid automaton to properly reflect the semantics of the system is a delicate process. 
Thus, unsurprisingly, several approaches aim at automatically constructing either hybrid automata or their simpler cousins, timed automata, based on execution traces of the system.
These traces are usually a development artifact as they get recorded during test runs.
Estimation methods, most prominently machine-learning, yield promising results in terms of reconstructing the correct discrete structure of the automaton and approximating the continuous dynamics.
Their great accuracy notwithstanding, the \emph{direction} of the approximation is unclear, resulting in over- and under-approximation. 
While this suffices for conveying the gist of the system, it does limit its capabilities in terms of safety-critical analyses.
For this reason, we propose a construction algorithm for \emph{conservative} hybrid automata, \ie, a guaranteed over-approximation of the original system.

Rather than relying solely on execution traces, the construction employs another development artifact:  a formal runtime monitoring specification.
This specification is crucial for safety-critical systems as it enables dynamic verification methods --- \emph{runtime verification} --- in which a dedicated component monitors the behavior of the system during the execution.
When the monitor deems this information indicative of a malfunction, it terminates the execution.
Note that the constraints imposed by the specification encompass the entire execution and change depending on the state of the system.
Hence, to judge the situation accurately, the monitor --- and by proxy the specification --- needs to keep track of different operational phases.

This comes to show that both the execution traces and the specification manifest expertise acquired during the development process: 
a) keeping track of operational phases requires the specification to contain information regarding the discrete control structure and b) since the execution traces are recorded test runs, they are expected to satisfy relevant coverage criteria.
Hence, not only do they cover the ``average'', expected behavior including initialization and termination, they also cover the extreme and corner case behavior.
Consequently, both the specification and the traces constitute indispensable resources. 

The construction algorithm proposed in this paper first extracts the implicit discrete control structure from the specification, resulting in a strong over-approximation of the system.
It then proceeds by a) refining this information based on discrete control signals contained in the execution traces and b) enriching it by analyzing the evolution of samples over time.
The resulting automaton is an under-approximation. 
Hence, the last step of the algorithm merges control modes within the automaton based on discrete evidence found in the traces to finally obtain an over-approximation of the original system.
This mixture of a top-down and bottom-up construction results in an automaton that is both a provable over-approximation and retains a high level of precision.

Apart from being conservative, the construction distinguishes itself from existing approaches in two major ways.
First, the specification roughly indicates the general discrete structure of the constructed automaton.
This alleviates the need to second-guess the structure in its entirety, reducing revisions to local sub-structures.
This pushes scalability far beyond $L^\ast$-based approaches~\cite{Angluin} like Medhat~\etal~\cite{MiningFramework} in which significant time is spent to determine the discrete structure.
Secondly, the construction reduces the level of over-approximation by merging modes of an under-approximation only if needed.
This can result in more fine-grained refinements than when successively widening dynamics until the language of the automaton encompasses every input trace~\cite{synthesis}.

An empirical evaluation validates three major claims.
First, the construction requires few traces to produce decent results. 
For a fourteen-mode automaton, for example, seven hand-picked or on average 35 random traces suffice for a perfect reconstruction.
Secondly, the precision --- while not flawless --- comes close to the optimal result for adequate input data.
Thirdly, the construction algorithm scales extraordinarily well.  
Even large automata with over 1000 modes can be constructed within mere seconds.
All three benefits are the result of relying on development artifacts in form of test traces and a runtime monitoring specification: a readily available resource often left under-utilized.

All in all, the contributions of this paper are:
\begin{itemize}
  \item A three step construction for conservative hybrid automaton:  
  First, it extracts a discrete over-approximation of the system from a runtime monitoring specification.
  This approximation is successively enriched with continuous information and refined into an under-approximation by incorporating data from execution traces.
  Lastly, it merges modes based on discrete evidence found in test traces until obtaining an over-ap\-prox\-i\-ma\-tion.
  \item A correctness proof of the construction:  under realistic assumptions on the input traces, the constructed automaton subsumes the language of the original system projected onto the behavior exposed through the set of input traces.
    The projection is required to eliminate parts of the system that are not exposed to the outside such as unreachable modes.
  \item Experimental results showcasing the quality and scalability of the conservative construction.
    Even small sets of input traces stemming from random walks allow for precise constructions.
    Moreover, automata with thousands of modes can be constructed in a matter of seconds.
\end{itemize}





%% file: source/problem.tex
\begin{figure*}[t]
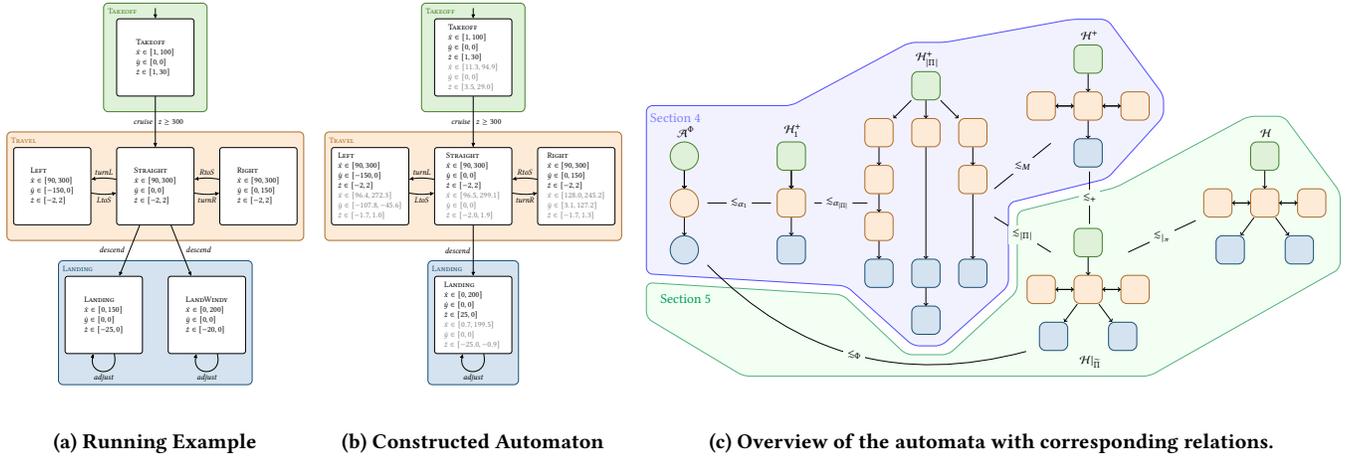

	\begin{subfigure}[b]{0.225\linewidth}
		\centering
		\include{source/figures/ha_drone_copy}
		\caption{Running Example}
		\label{fig:ha:drone}
	\end{subfigure}
	\hfill
	\begin{subfigure}[b]{0.225\linewidth}
		\centering
		\include{source/figures/ha_drone_merged}
		\caption{Constructed Automaton}
		\label{fig:ha:drone_output}
	\end{subfigure}
	\hfill
	\begin{subfigure}[b]{.525\linewidth}
\centering
\include{source/figures/overview}
\caption{Overview of the automata with corresponding relations.}	
\label{fig:overview}
\end{subfigure}
\caption{
The specification automaton for an aircraft superimposed by its hybrid automaton in \Cref{fig:ha:drone} and the result of the construction from three hand-picked traces (black dynamics) and ten random traces (gray dynamics) in \Cref{fig:ha:drone_output}.
An overview of the automata is displayed in \Cref{fig:overview}, where he notation $\aut_1 \lesssim \aut_2$ indicates that $\aut_2$ simulates $\aut_1$.}
\end{figure*}

\section{Motivation}\label{sec:motivation}
The foundation for the conservative construction is a set of traces generated from an unknown system and a runtime monitoring specification thereof.
As a running example, consider the automaton depicted in \Cref{fig:ha:drone}, a simple model for an aircraft.
The system starts in a takeoff mode and resides there until reaching cruising altitude.
Here, it can go straight or adjust its course via left and right curves until it attempts a landing.
Under windy conditions, the descend-phase is elongated as a precaution.

A specification for the aircraft imposes several constraints depending on the current state of the system.
For example, during takeoff, the specification requires the aircraft to accelerate; while traveling it requires a stable altitude; during landing it requires the landing gear to be lowered.
An analysis of the specification hence yields a state machine with coarse information on different execution phases as well as conditions on phase-changes.
The state machine is depicted in color in~\Cref{fig:ha:drone}, superimposed by the aircraft automaton.
As can be seen, the specification does not distinguish between maintaining course or adjusting it; the requirements on the system remain the same.
Yet, it contains no indication regarding the continuous behavior of the system.

To fill these gaps, the conservative construction then successively enriches the specification automaton with information extracted from the set of traces.
The whole process is illustrated in \Cref{fig:overview}.
It first translates the specification automaton~$\finaut^\spec$ into a hybrid automaton~$\constraut_1$.
For each step of each trace, it adds more modes into the automaton while maintaining the structure provided by the specification.
The result, \ie, $\constraut_{\card\traces}$, is by design overly restrictive: it consists of a single, isolated path for each trace.
A merge process based on the discrete behavior of modes remedies this problem and finally produces~$\constraut$.

The key point behind~$\constraut$ is that it is \emph{conservative}, \ie, under certain assumptions on the specification and traces, $\constraut$ over-ap\-prox\-i\-mates $\aut$.
The assumptions are three-fold: the specification needs to a) be a coarse abstraction of the actual system, b) agree with the system on phase changes, and c) the set of traces needs to encompass sufficient information on the discrete behavior.
While these assumptions seem strong, they are tailored for the use case at hand such that they are expected to be satisfied naturally.
The first and second assumption concern the specification, which is hand-crafted specifically for the system.
Hence, the specifier must have had knowledge regarding the abstract control structure, \eg through which phases the system traverses during a mission.
Evidently, the concrete control structure refines the abstract structure for more fine-grained control.
This abstract control structure manifests itself in the specification.
Here, each mode of the abstract structure imposes a different set of requirements on the system.
A violation of such a requirement constitutes a safety-hazard.
Hence, engineers need to define precisely when the requirements on the system changes.
As a result, the specification needs to contain the same precise criterion for a phase change as the system itself.

Regarding the third criterion, recall that the set of traces is a development artifact that arose from the testing process.
Thus, not only do they cover large parts of the system's regular execution, they represent executions in which the system was purposefully coerced into extreme and corner case behavior.
Moreover, the empirical evaluation revealed that even small sets of random walks through the underlying system already produces an adequate trace set.  
Here, ``small'' means on average 32 traces for a fourteen state automaton and two traces for a seven state automaton.
This reasoning shows why these assumptions, while strong, are realistic when considering carefully engineered, safety-critical system.
A formal description of the assumptions and the proof that~$\constraut$ is conservative can be found in \Cref{sec:correctness}.

The empirical evaluation of the approach in \Cref{sec:eval} allows for validating three claims.
\begin{enumerate*}[label=\alph*)]
\item The construction requires a \emph{low volume of input traces} --- especially when compared to machine-learning approaches. 
For the running example, the construction requires as little as three traces of length eight.
\item It \emph{scales linearly} for increasing dimension and \emph{quadratically} in the number of traces and size of the original/constructed system. 
Constructing a three-dimensional automaton with $2^{10}$ modes based on a specification with nine states and 512 traces requires less than a second.
\item Though over-approximating, the constructed automaton has a \emph{high level of precision}.
    For the aircraft, the construction is identical to the original apart from a merge of the two terminal modes.
\end{enumerate*}

%% file: source/figures/ha_drone_copy.tex
\resizebox{\linewidth}{!}{
\begin{tikzpicture}[->,>=stealth',shorten >= 1pt,auto]
\node[state, fill=white, initial above, initial text=] (takeoff) at (0,0){%
  \modename{Takeoff}\\
  $\dot{x} \in [1,100]$ \\
  $\dot{y} \in [0,0]$ \\
  $\dot{z} \in [1,30]$
};

\node[state, fill=white] (straight) at (0, -5) {%
  \modename{Straight}\\
  $\dot{x} \in [90,300]$ \\
  $\dot{y} \in [0,0]$ \\
  $\dot{z} \in [-2,2]$ 
};

\node[state, fill=white] (left) at (-4, -5) {%
  \modename{Left} \\
  $\dot{x} \in [90,300]$ \\
  $\dot{y} \in [-150,0]$ \\
  $\dot{z} \in [-2,2]$ 
};

\node[state, fill=white] (right) at (4, -5) {%
  \modename{Right} \\
  $\dot{x} \in [90,300]$ \\
  $\dot{y} \in [0,150]$ \\
  $\dot{z} \in [-2,2]$
};

\node[state, fill=white] (landing) at (-2, -10) {%
  \modename{Landing}\\
  $\dot{x} \in [0,150]$\\
  $\dot{y} \in [0,0]$ \\
  $\dot{z} \in [-25,0]$
};

\node[state, fill=white] (landwindy) at (2, -10) {%
  \modename{LandWindy} \\
  $\dot{x} \in [0,200]$ \\
  $\dot{y} \in [0,0]$ \\
  $\dot{z} \in [-20,0]$ 
};

\begin{scope}[-,on background layer]
  \filldraw[fill=prefix!20, draw=prefix!70!black, rounded corners=5pt]
    ($(takeoff.north west)+(-0.5,0.6)$)  -- ($(takeoff.north east)+(0.5,0.6)$) -- ($(takeoff.south east)+(0.5,-0.6)$) -- ($(takeoff.south west)+(-0.5,-0.6)$) -- cycle node[anchor=north west] at ($(takeoff.north west)+(-0.45,0.5)$) {$\color{prefix!70!black} \textsc{Takeoff}$};
\end{scope}

\begin{scope}[-,on background layer]
  \filldraw[fill=prelude!20, draw=prelude!70!black, rounded corners=5pt]
    ($(left.north west)+(-0.25,0.6)$) -- ($(right.north east)+(0.25,0.6)$) -- ($(right.south east)+(0.25,-0.6)$) -- ($(left.south west)+(-0.25,-0.6)$) -- cycle node[anchor=north west] at ($(left.north west)+(-0.2,0.5)$) {\color{prelude!70!black} \textsc{Travel}};
\end{scope}

\begin{scope}[-,on background layer]
  \filldraw[fill=loop!20, draw=loop!70!black, rounded corners=5pt]
    ($(landing.north west)+(-0.25,0.6)$) -- ($(landwindy.north east)+(0.25,0.6)$) -- ($(landwindy.south east)+(0.25,-1.2)$) -- ($(landing.south west)+(-0.25,-1.2)$) -- cycle node[anchor=north west] at ($(landing.north west)+(-0.2,0.5)$) {\color{loop!70!black} \textsc{Landing}};
\end{scope}

\path (takeoff) edge node[left] {%
        $\mathit{cruise}$
      } node[right] {%
        $z \geq 300$
      } (straight)
      (straight) edge[bend right=10] node[above] {%
        $\mathit{turnL}$
      } (left)
      (straight) edge[bend right=10] node[below] {%
        $\mathit{turnR}$
      } (right)
      (left) edge[bend right=10] node[below] {%
        $\mathit{LtoS}$
      } (straight)
      (right) edge[bend right=10] node[above] {%
        $\mathit{RtoS}$
      } (straight)
      (straight) edge node[above left,pos=.6] {%
        $\mathit{descend}$
      } (landing)
      (straight) edge node[above right,pos=.6] {%
        $\mathit{descend}$
      } (landwindy)
      (landing) edge[loop below,looseness=3] node[below] {%
        $\mathit{adjust}$
      } (landing)
      (landwindy) edge[loop below,looseness=3] node[below] {%
        $\mathit{adjust}$
      } (landwindy);
\end{tikzpicture}
}

%% file: source/figures/ha_drone_merged.tex
\resizebox{\linewidth}{!}{
\begin{tikzpicture}[->,>=stealth',shorten >= 1pt,auto]
\node[state, fill=white, initial above, initial text=] (takeoff) at (0,0) {%
  \modename{Takeoff}\\
  $\dot{x} \in [1,100]$ \\
  $\dot{y} \in [0,0]$ \\
  $\dot{z} \in [1,30]$ \\
  \color{gray} $\dot{x} \in [11.3,94.9]$ \\
  \color{gray} $\dot{y} \in [0,0]$ \\
  \color{gray} $\dot{z} \in [3.5,29.0]$
};

\node[state, fill=white] (straight) at (0, -5) {%
  \modename{Straight}\\
  $\dot{x} \in [90,300]$ \\
  $\dot{y} \in [0,0]$ \\
  $\dot{z} \in [-2,2]$ \\
  \color{gray} $\dot{x} \in [96.5,299.1]$ \\
  \color{gray} $\dot{y} \in [0,0]$ \\
  \color{gray} $\dot{z} \in [-2.0,1.9]$
};

\node[state, fill=white] (left) at (-4, -5) {%
  \modename{Left} \\
  $\dot{x} \in [90,300]$ \\
  $\dot{y} \in [-150,0]$ \\
  $\dot{z} \in [-2,2]$ \\
  \color{gray} $\dot{x} \in [96.4,272.3]$ \\
  \color{gray} $\dot{y} \in [-107.8,-45.6]$ \\
  \color{gray} $\dot{z} \in [-1.7,1.0]$
};

\node[state, fill=white] (right) at (4, -5) {%
  \modename{Right} \\
  $\dot{x} \in [90,300]$ \\
  $\dot{y} \in [0,150]$ \\
  $\dot{z} \in [-2,2]$ \\
  \color{gray} $\dot{x} \in [128.0,245.2]$ \\
  \color{gray} $\dot{y} \in [3.1,127.2]$ \\
  \color{gray} $\dot{z} \in [-1.7,1.3]$
};

\node[state, fill=white] (landing) at (0, -10) {%
  \modename{Landing} \\
  $\dot{x} \in [0,200]$ \\
  $\dot{y} \in [0,0]$ \\
  $\dot{z} \in [25,0]$ \\
  \color{gray} $\dot{x} \in [0.7,199.5]$ \\
  \color{gray} $\dot{y} \in [0,0]$ \\
  \color{gray} $\dot{z} \in [-25.0,-0.9]$
};

\begin{scope}[-,on background layer]
  \filldraw[fill=prefix!20, draw=prefix!70!black, rounded corners=5pt]
    ($(takeoff.north west)+(-0.5,0.6)$)  -- ($(takeoff.north east)+(0.5,0.6)$) -- ($(takeoff.south east)+(0.5,-0.6)$) -- ($(takeoff.south west)+(-0.5,-0.6)$) -- cycle node[anchor=north west] at ($(takeoff.north west)+(-0.45,0.5)$) {$\color{prefix!70!black} \textsc{Takeoff}$};
\end{scope}

\begin{scope}[-,on background layer]
  \filldraw[fill=prelude!20, draw=prelude!70!black, rounded corners=5pt]
    ($(left.north west)+(-0.25,0.6)$) -- ($(right.north east)+(0.25,0.6)$) -- ($(right.south east)+(0.25,-0.6)$) -- ($(left.south west)+(-0.25,-0.6)$) -- cycle node[anchor=north west] at ($(left.north west)+(-0.2,0.5)$) {\color{prelude!70!black} \textsc{Travel}};
\end{scope}

\begin{scope}[-,on background layer]
  \filldraw[fill=loop!20, draw=loop!70!black, rounded corners=5pt]
    ($(landing.north west)+(-0.25,0.6)$) -- ($(landing.north east)+(0.25,0.6)$) -- ($(landing.south east)+(0.25,-1.2)$) -- ($(landing.south west)+(-0.25,-1.2)$) -- cycle node[anchor=north west] at ($(landing.north west)+(-0.2,0.5)$) {\color{loop!70!black} \textsc{Landing}};
\end{scope}

\path (takeoff) edge node[left] {%
        $\mathit{cruise}$
      } node[right] {%
        $z \geq 300$
      } (straight)
      (straight) edge[bend right=10] node[above] {%
        $\mathit{turnL}$
      } (left)
      (straight) edge[bend right=10] node[below] {%
        $\mathit{turnR}$
      } (right)
      (left) edge[bend right=10] node[below] {%
        $\mathit{LtoS}$
      } (straight)
      (right) edge[bend right=10] node[above] {%
        $\mathit{RtoS}$
      } (straight)
      (straight) edge node[left] {%
        $\mathit{descend}$
      } (landing)
      (landing) edge[loop below,looseness=3] node[below] {%
        $\mathit{adjust}$
      } (landing);
\end{tikzpicture}
}

%% file: source/figures/overview.tex
\resizebox{\linewidth}{!}{
\tikzstyle{state}=[draw=none, rectangle, fill=none, align=left, minimum width=2cm, 
minimum height = 2.3cm, rounded corners=1mm, inner sep=0pt]
\begin{tikzpicture}[
specgreen/.style={circle, thick, draw = prefix!70!black, fill=prefix!20, minimum size= .8cm},	
specorange/.style={circle, thick, draw = prelude!70!black, fill=prelude!20, minimum size= .8cm},	
specblue/.style={circle, thick, draw = loop!70!black, fill=loop!20, minimum size= .8cm},	
hybridgreen/.style={rectangle, thick, draw = prefix!70!black, fill=prefix!20, rounded corners=5pt, minimum size= .8cm},	
hybridorange/.style={rectangle, thick, draw = prelude!70!black, fill=prelude!20, rounded corners=5pt, minimum size= .8cm},	
hybridblue/.style={rectangle, thick, thick, draw = loop!70!black, fill=loop!20, rounded corners=5pt, minimum size= .8cm},	
hybedge/.style={->}
]
\node[state, thick] (specautomaton) [label = above:$\finaut^\spec$]{
	\begin{tikzpicture}
		\node[specgreen](top){};
		\node[specorange](mid) [below = 0.5 of top]{};
		\node[specblue](bot) [below = 0.5 of mid]{};
		
		\path
		(top) edge[hybedge, ->] (mid)
		(mid) edge[hybedge,->] (bot)
		(top) edge[hybedge,->] (mid)
		;
	\end{tikzpicture}
};

\node[state, thick] (oneautomaton) [right = 1 of specautomaton, label = above:$\constraut_{1}$] {
\begin{tikzpicture}
		\node[hybridgreen](top){};
		\node[hybridorange](mid) [below = 0.5 of top]{};
		\node[hybridblue](bot) [below = 0.5 of mid]{};
		
		\path
		(top) edge[hybedge,->] (mid)
		(mid) edge[hybedge,->] (bot)
		(top) edge[hybedge,->] (mid)
		;
	\end{tikzpicture}
};

\node[state, thick] (traceautomaton) [right = 1 of oneautomaton, label = above:{$\constraut_{|\traces|}$}] {
	\begin{tikzpicture}
		\node[hybridgreen](top){};
		\node[hybridorange](midmid) [below = 0.5 of top]{};
		\node[hybridorange](midleft) [left = 0.5 of midmid]{};
		\node[hybridorange](midleft1) [below = 0.5 of midleft]{};
		\node[hybridorange](midleft2) [below = 0.5 of midleft1]{};
		\node[hybridorange](midright1) [right = 0.5 of midmid]{};
		\node[hybridorange](midright2) [below = 0.5 of midright1]{};
		\node[hybridblue](botleft) [below = 0.5 of midleft2]{};
		\node[hybridblue](botmid) [ right= 0.5 of botleft]{};
		\node[hybridblue](botmid1) [ below= 0.5 of botmid]{};
		\node[hybridblue](botright) [ right= 0.5 of botmid]{};

		\path
		(top) edge[hybedge,->] (midmid)
		(midmid) edge[hybedge,->] (botmid)
		(botmid) edge[hybedge,->] (botmid1)
		(top) edge[hybedge,->] (midleft)
		(midleft) edge[hybedge,->] (midleft1)
		(midleft1) edge[hybedge,->] (midleft2)
		(midleft2) edge[hybedge,->] (botleft)
		(top) edge[hybedge,->] (midright1)
		(midright1) edge[hybedge,->] (midright2)
		(midright2) edge[hybedge,->] (botright)
		;
	\end{tikzpicture}
};

\node[state, thick] (constraut) [above right = -1 and 1 of traceautomaton, label=above:{$\constraut$}, yshift=-50pt] {
\begin{tikzpicture}
		\node[hybridgreen](top){};
		\node[hybridorange](midmid) [below = 0.5 of top]{};
		\node[hybridorange](midleft) [left = 0.5 of midmid]{};
		\node[hybridorange](midright) [right = 0.5 of midmid]{};
		\node[hybridblue](bot) [below = 0.5 of midmid]{};
		
		\path
		(top) edge[hybedge,->] (midmid)
		(midmid) edge[hybedge,->] (bot)
		(midleft) edge[hybedge,->] (midmid)
		(midmid) edge[hybedge,->] (midleft)
		(midmid) edge[hybedge,->] (midright)
		(midright) edge[hybedge,->] (midmid)
		;
	\end{tikzpicture}
};

\node[state, thick] (projaut) [below right = -2 and 1 of traceautomaton, label = below:{$\projection{\aut}$}, yshift=30pt] {
	\begin{tikzpicture}
		\node[hybridgreen](top){};
		\node[hybridorange](midmid) [below = 0.5 of top]{};
		\node[hybridorange](midleft) [left = 0.5 of midmid]{};
		\node[hybridorange](midright) [right = 0.5 of midmid]{};
		\node[hybridblue](botleft) [below left= 0.5 and 0.15 of midmid]{};
		\node[hybridblue](botright) [below right= 0.5 and 0.15 of midmid]{};
		
		\path
		(top) edge[hybedge,->] (midmid)
		(midmid) edge[hybedge,->] (botleft)
		(midmid) edge[hybedge,->] (botright)
		(midleft) edge[hybedge,->] (midmid)
		(midmid) edge[hybedge,->] (midleft)
		(midmid) edge[hybedge,->] (midright)
		(midright) edge[hybedge,->] (midmid)
		;
	\end{tikzpicture}

};

\node[state, thick] (aut) [right = 6 of traceautomaton, label = above:{$\aut$}] {
	\begin{tikzpicture}
		\node[hybridgreen](top){};
		\node[hybridorange](midmid) [below = 0.5 of top]{};
		\node[hybridorange](midleft) [left = 0.5 of midmid]{};
		\node[hybridorange](midright) [right = 0.5 of midmid]{};
		\node[hybridblue](botleft) [below left= 0.5 and 0.15 of midmid]{};
		\node[hybridblue](botright) [below right= 0.5 and 0.15 of midmid]{};
		
		\path
		(top) edge[hybedge,->] (midmid)
		(midmid) edge[hybedge,->] (botleft)
		(midmid) edge[hybedge,->] (botright)
		(midleft) edge[hybedge,->] (midmid)
		(midmid) edge[hybedge,->] (midleft)
		(midmid) edge[hybedge,->] (midright)
		(midright) edge[hybedge,->] (midmid)
		;
	\end{tikzpicture}
};

\node[draw = none](discretesim) at ($(specautomaton)!0.5!(oneautomaton)$){$\lesssim_{\alpha_{1}}$};
\node[draw = none](onetraces) at ($(oneautomaton)!0.33!(traceautomaton)$){$\lesssim_{\alpha_{|\traces|}}$};

\path
($(specautomaton.east) + (-0.4,0)$) edge[-, thick] node[below]{} (discretesim)
(discretesim) edge[-, thick] node[below]{} ($(oneautomaton.west) + (0.4, 0)$)
($(oneautomaton.east) + (-0.4,0)$) edge[-, thick] node[below]{} (onetraces)
(onetraces) edge[-, thick] node[below]{} ($(traceautomaton.west) + (0.1, 0)$)
($(specautomaton.south east) + (-0.4,0)$) edge[-, bend right, thick] node[fill=green!5]{$\lesssim_{\Phi}$} (projaut.south west)
($(traceautomaton.east) + (0.1, 0.4)$) edge[-, thick] node[fill=blue!5]{$\lesssim_{M}$} ($(constraut.south west) + (0.7, 0.7)$)
($(traceautomaton.east) + (0.1, -0.4)$) edge[-, thick] node[fill=green!5]{$\lesssim_{|\traces|}$} ($(projaut.north west) + (+0.7, -0.7)$)
($(projaut.north east) + (-0.7, -0.7)$) edge[-, thick] node[fill=green!5]{$\lesssim_{|_{\pi}}$} ($(aut.west) + (-0.2, -0.5)$)
($(projaut.north) + (0,0.1)$) edge[-, thick] node[fill=green!5]{$\lesssim_{+}$} ($(constraut.south) + (0,-0.1)$)
;

\begin{scope}[-,on background layer]
  \filldraw[draw = blue!70, fill=blue!5, rounded corners=5pt]
    ($(specautomaton.north west)+(-0.1,1)$) -- node[above right = 1.8cm and 0cm of specautomaton] {\Large\color{blue!50}Section 4} 
    ($(specautomaton.south west)+(-0.1,-0.25)$) --($(traceautomaton.south west)+(-0.5,1.4)$) -- ($(traceautomaton.south)+(-0.4,-0.3)$)  -- ($(traceautomaton.south)+(+0.4,-0.3)$) -- ($(traceautomaton.south east)+(+0.5,1.4)$)  -- ($(traceautomaton.east)+(+0.5,0)$) --($(constraut.south)+(+0.5,-0.3)$)-- ($(constraut.east)+(+0.3,-0.3)$) 
  -- ($(constraut.east)+(+0.3,+0.3)$)  --($(constraut.north)+(+0.5,.7)$)--($(constraut.north)+(-0.5,.7)$)
    --($(traceautomaton.north)+(+0.5,.8)$) -- ($(traceautomaton.north)+(-0.5,.8)$)
    -- ($(specautomaton.north west)+(4,1)$) -- cycle;
\end{scope}

\begin{scope}[-,on background layer]
  \filldraw[draw = ForestGreen!70, fill=green!5, rounded corners=5pt]
  ($(specautomaton.south west)+(-0.1,-0.5)$) -- ($(traceautomaton.south west)+(-0.6,1.15)$) -- 
  ($(traceautomaton.south)+(-0.6,-0.55)$)--
  ($(traceautomaton.south)+(+0.6,-0.55)$) -- ($(traceautomaton.south east)+(+0.75,1.25)$)
    -- ($(traceautomaton.east)+(+0.75,-0.25)$) -- ($(constraut.south)+(+0.75,-0.55)$) --($(aut.north)+(-0.5,+0.7)$)-- ($(aut.north)+(+0.5,+0.7)$) 
  -- ($(aut.east)+(+0.3,+0.3)$) --($(aut.south east)+(+0.3, -0.1)$) -- ($(projaut.south east)+(-0.1,-.7)$) -- ($(projaut.south west)+(-8,-.7)$) --($(specautomaton.south west)+(-0.1,-1.5)$) -- cycle node[below = .7cm of specautomaton]{\Large\color{ForestGreen}Section 5};
\end{scope}

\end{tikzpicture}
}

%% file: source/preliminaries.tex
\section{Preliminaries}\label{sec:prelims}

\begin{definition}[Interval]
  $\interval^n = \real^{2n}$ denotes an $n$-dimensional rectangle where $\interval^1 = \interval$ is an interval over $\real$. The multiplication of a $k$- and an $\ell$-dimensional rectangle yields a $(k+\ell)$-dimensional rectangle.  Addition and multiplication of intervals with scalars are geometric translation and scaling, respectively.
\end{definition}


\begin{definition}[Convex Hull]
Let $S$ be a convex set.
The \emph{convex hull} of two convex sets $A \subseteq S$ and $B \subseteq S$ is the minimal convex set covering both $A$ and~$B$.
Further, let $\mathcal{S}$ be a non-empty set of convex sets. 
$\Convex{\mathcal{X}}$ computes the convex set covering all elements of $\mathcal{X}$ by applying the convex hull iteratively in arbitrary order. 
$\Convex{\mathcal{X}}$ yields the identity for singleton inputs.
\begin{align*}
\chull{A, B} &= \bigcup_{p_a \in A}\bigcup_{p_b \in B} [p_a, p_b] \\
\Convex{Y \cup \mathcal{X}} &= \chull{Y, \Convex{\mathcal{X}}}
\end{align*}
\end{definition}


\subsection{Hybrid Automata}\label{sec:hybrid:automata}

An $n$-dimensional \textit{Multi-Rectangular Hybrid Automaton} $\aut$ is a 6-tuple $(\modes, \actions, \allowbreak \flow, \edges, \guard, \initialstate)$ over $\rn$. 
$\modes$ denotes the finite set of discrete control modes.
A \emph{state} $\sstate \in \modes \times \rn$ of the automaton consists of a discrete mode and a continuous state. 
Here, $\initialstate = (\initial{\mode}, \initial{x}) \in M \times \rn$ is the \emph{initial state} and $\actions$ is the finite set of \emph{action labels}.
$\flow \from \modes \to \interval^n$ defines the \emph{dynamics} of modes.  
When entering one, a random value is drawn from the respective interval for each dimension.
$\edges \subseteq \modes \times \actions \times \modes$ is the finite set of edges containing discrete, labeled \emph{transitions} between two modes.
Lastly, $\guard \from \edges \to \interval^n$ assigns guard conditions to edges. A transition can only be taken in a state if its continuous component lies within the rectangle.
\Cref{fig:ha:drone} shows the easy-to-grasp visual representation of a hybrid automaton where each white rectangle constitutes a mode. 
The text inside it defines the dynamics as differential equations, edge labels state guard conditions.

\paragraph{Semantics.}

Hybrid automata allow for two kinds of transitions: control mode changes according to $\edges$ and delays according to the $\flow$ of the current mode during which the system state evolves continuously. 
More formally, the semantics of a multi-rectangular hybrid automaton \aut are defined based on valid \emph{omniscient traces} through \aut.
An omniscient trace $\godtrace \in \rn \times \rplus \times \rn \times (\edges \times \rplus \times \rn)^k$ with $\godtrace = x_0, \delay_0, x_1, e_1, \delay_1, \dots, x_k, e_k, \delay_k, x_{k+1}$ is valid for an automaton ($\godtrace \tracecompliant \aut$) iff:
\begin{enumerate*}[label=\alph*)]
\item 
the trace starts in the initial state of \aut, \ie, $\initial{x} = x_0$.
\item 
the first discrete transition starts in the initial mode, so $e_1 = (\initial{\mode}, \action, \mode)$ for some $\action$ and $\mode$.
\item 
all guards are satisfied: $\forall 1 \leq i \leq k\colon x_i \in \guard(e_i)$.
\item 
all delay transitions are valid, \ie, for $0 \leq i \leq k$ and $e_{i+1} = (\mode_{s}, \action, \mode_{t})$, the state changes according to the flow: \(x_{i+1} \in (x_{i} + \delay_{i}\cdot\flow(\mode_s)) \).
\end{enumerate*}
The language of an automaton is the set of valid traces: $\lang\aut = \Set{\godtrace \given \godtrace \tracecompliant \aut}$. 

\begin{definition}[Observable Traces]
  An observable trace $\trace$ is an omniscient trace stripped of its information regarding source and target modes: 
  \(\trace \in \rn \times \rplus \times \rn \times (\actions \times \rplus \times \rn)^k\).
\end{definition}
  For the remainder of the paper, unless stated otherwise, a trace refers to an observable trace and we only consider finite languages.

\paragraph{Notation.}
Any automaton with decoration such as $\constraut$ will be implicitly destructed into its components with the same decoration, \eg $\modes^+$ denotes the modes of $\constraut$.
$\dimaccess{x}{k}$ denotes the $k$th component of the $n$-dimensional vector $x$ for $0< k \leq n$.  
For a finite set $A$, $\card{A}$ denotes the cardinality of $A$.
The length $\card\trace$ of a trace $\trace \in \rn \times \rplus \times \rn \times (\actions \times \rplus \times \rn)^k$ is the number of timed transitions occurring in it, \ie, $k+1$. 
A trace of length $k+1$ is implicitly destructed into the following components: $\trace = x^\trace_0, \delay^\trace_0, x^\trace_1, e^\trace_1, \delay^\trace_1, \dots, x^\trace_k, e^\trace_k, \delay^\trace_k, x^\trace_{k+1}$.
Moreover, mode $\mode$ is a member of an omniscient trace if it reaches the mode at least once: $\mode\in\trace \iff \exists i\colon \edge^\godtrace_i = (\mode_1, \action, \mode_2)$ with \(\mode \in \Set{\mode_1,\mode_2}\).
A \emph{step} of a trace is the combination of a delay and a discrete transition.
Further, let $\traces$ be a sequence of traces in arbitrary order.  
Then, $\trace_i$ denotes the $i$th entry of the sequence with $i \leq \card\traces$.

\paragraph{Bisimulation.}
Discrete bisimulation on two (hybrid) automata is defined conventionally by disregarding any continuous behavior or behavior not shared among the automata.
\begin{definition}[Discrete Bisimulation]
  Two modes of $\mode_1, \mode_2$ of two automata $\aut_1$, $\aut_2$ are discretely bisimilar $\mode_1 \discretebisim \mode_2$ iff for all transition labels $\action \in \actions_1 \cap \actions_2$:
  \[ (\mode_1, \action, \mode_1') \in \edges_1 \implies \exists \mode_2'\colon (\mode_2, \action, \mode_2') \in \edges_2 \land \mode_1' \discretebisim \mode_2' \text{ and} \]
  \[ (\mode_2, \action, \mode_2') \in \edges_2 \implies \exists \mode_1'\colon (\mode_1, \action, \mode_1') \in \edges_1 \land \mode_2' \discretebisim \mode_1' \text{\hphantom{ and}}\]
  Further, the two automata are discretely bisimilar $\aut_1 \discretebisim \aut_2$ iff $\initial{\mode^1} \discretebisim \initial{\mode^2}$.
\end{definition}

%% file: source/construction.tex
\section{Constructing Conservative Automata}\label{sec:construction}

The construction proceeds in three steps:
First, it extracts information from the specification to obtain a finite state machine $\finaut^\spec$ and a table mapping discrete control mode changes to conditions for undergoing such a change.
The automaton is a coarse abstraction of the underlying system.
The second step transforms it into a hybrid automaton $\constraut$ and iteratively refines it by extracting information regarding the continuous behavior from the input traces.
By design, the refinement overshoots its goals, resulting in an abstraction that is too fine.
As a remedy, the third step merges parts of the automaton to construct a conservative automaton.
The alternation in coarseness has the effect that the resulting automaton is an over-approximation without being overly permissive.

\subsection{Extracting Discrete Information from the Specification}
The requirements on the system change depending on its current state.
For example, during the landing of an airplane, the landing gear must be lowered whereas it is required to be retracted when on traveling altitude.
Hence, the specification needs to keep track of relevant parts of the system state to impose the proper restrictions.
This process of keeping track induces an abstract state machine that lacks any information on the continuous dynamics since the monitor solely relies on external input data such as sensor readings.
Each abstract state may summarize several concrete modes of the actual system.  
In the plane example, the requirements on the abstract mode ``in full flight'' apply to both the control mode ``no wind'' and ``tail wind'' even though they have different continuous dynamics.
By assumption, the contrary is false:  a change of requirements on the system is always accompanied by a change in concrete modes. 
Intuitively, a change of requirements is strongly linked to an action or reaction of the system: approaching a geofence imposes new constraints and hence prompts a reactionary change of course to satisfy them.
A formalization of these assumptions follows in \Cref{sec:inputreqs}.

The extraction of the abstract automaton varies depending on the specification language.
This paper uses the \rtlola \cite{RTLola,LolaTutorial} monitoring framework since it was designed for and integrated into safety-critical cyber-physical systems~\cite{StreamLAB,lolaindustrial,uav}. 
An \rtlola specification consists of input streams representing data the monitor receives from the system, output streams and triggers. 
With output streams the specifier declares how to process input data with the goal of analyzing the state of the system.
Lastly, triggers define conditions on output streams. 
The satisfaction of a trigger condition prompts the monitor to emit a message to the system, informing it about a violation of a safety requirement or the detection of a phase change.

A specification can keep track of the current set of requirements imposed on the system by using an output stream $\mode^\spec$.
The value of $\mode^\spec$ indicates in which abstract state the system is.
Assume there are two abstract states $\mode_1^\spec$ and $\mode_2^\spec$, and a state transition occurs under some condition $\phi$.
Then, the $\mode^\spec$ stream has the following shape:
\begin{lstlisting}
output $\mode^\spec$ := if $\mode^\spec$ = $\mode_1^\spec$ $\land$ $\phi$ then $\mode_2^\spec$ else last($\mode^\spec)$
\end{lstlisting}
Here, \lstinline{last($\mode^\spec$)} provides the last value of the $\mode^\spec$.
For more possible abstract states and transitions, the conditional statement can be extended accordingly.
In addition, a state change is accompanied by a respective trigger:
\begin{lstlisting}
trigger last($\mode^\spec$) = $\mode^\spec_1$ $\land$ $\mode^\spec$ = $\mode^\spec_2$ $\land$ $\phi$ "$\color{redstrings}\mode^\spec_1 \rightarrow \mode^\spec_1$ with $\color{redstrings}\action$."
\end{lstlisting}  
The trigger checks for a change in $\mode^\spec$ from $\mode_1^\spec$ to $\mode_2^\spec$ and emits this information coupled with the name $\action$ of the respective transition.

An analysis of the output stream and trigger declarations yields two artifacts:
\begin{definition}[Specification Automaton]\label{def:constr:spec}
Given a specification \spec, $\finaut^\spec = (V^\spec, \edges^\spec, \initial{v}^\spec)$ is the \emph{abstract specification automaton} of \spec and $\guardspec = (V^\spec\times\action\times V^\spec) \to \interval^n$ is the \emph{guard condition table} of \spec.
\end{definition}
In the construction of $\finaut^\spec$, $V^\spec$ and $\initial{v}^\spec$ are the domain and initial value of $\mode^\spec$, respectively.
Then, for each trigger as the one stated before, $\edges^\spec$ contains the edge $(\mode^\spec_1,\lambda,\mode^\spec_2)$ and $\guardspec(\edge) = \phi$.
Note that the following assumes $\finaut^\spec$ to be free of unreachable states and related edges.  
This is the case in sensible specifications and can easily be enforced by pruning the respective parts of the graph.

\subsection{Extracting Continuous Information from Traces}

While the specification provides information about the system's \emph{discrete} structure, the traces reveal how the \emph{continuous} state of the system evolves over time.
They also reveal mode changes within a single abstract state.
This information allows for transforming $\finaut^\spec$ into a more fine-grained automaton with annotated dynamics in each mode.
For this, the transformation iteratively constructs an automaton $\constraut$, processing each position of all traces in separation.
This requires to keep track of two maps: a concrete mode-map $\modemap\from\traces\to\modes$ that maps each trace to the mode of the constructed automaton in which it currently resides, and an abstract mode-map $\alpha\from\mode\to V^\spec$ mapping each concrete mode to an abstract one in the specification automaton $\finaut^\spec$. 
\begin{definition}[Construction Initialization]\label{def:constr:init}
The construction starts with a quasi-empty hybrid automaton $\constraut_1$ that is structurally similar to $\finaut^\spec$, a concrete mode-map~$\modemap_1$ and an abstract mode-map~$\alpha_1$ defined as:
\begin{equation*}
\begin{gathered}
  \modes_1 = \Set{\initial\mode} \quad 
  \actions_1 = \emptyset \quad
  \edges_1 = \emptyset \quad 
  \guard_1(e) = \neutralelem \quad
  \modemap_1(\trace) = \initial\mode \\
  \alpha_1(\initial\mode) = \initial{v^\spec} \quad
  \flow_1(\mode) 
  = \prod_{\trace_i \in \traces} \solve(x^{\trace_i}_0, x^{\trace_i}_1, \delay^\trace_0)
\end{gathered}
\end{equation*}
Here, $\neutralelem$ denotes the neutral element with respect to the multiplication of intervals. 
Moreover, the $\solve$-function computes the singular interval representing the linear dynamics exhibited by a delay transition:
\[ \solve(x, x', \delay) = [(x' - x)\delay^{-1}, (x' - x)\delay^{-1}] \]
Thus, $\constraut_1$ already incorporates the information of each trace regarding their first delay transition.
\end{definition}

After the initialization, the procedure successively incorporates information contained in further positions of the traces.
\begin{definition}[Construction Step]\label{def:constr:step}
Given the automaton $\constraut_k$, concrete mode-map $\modemap_k$, and abstract mode-map $\alpha_k$ from the previous construction step. 
Consider the $k$-th step of each input trace, \ie, $x^{\trace_i}_k$, $\action^{\trace_i}_k$, $\delay^{\trace_i}_k$, and $x^{\trace_i}_{k+1}$ for all $0 < i \leq \card\traces$.
The $k$th step of the construction produces $\constraut_{k+1}, \modemap_{k+1}$, and $\alpha_{k+1}$.

For brevity, given $\edge = (\mode_1, \action, \mode_2)$, let $\alpha_k(\edge) = (\alpha_k(\mode_1), \action, \alpha_k(\mode_2))$.
Moreover, let $\spec(\trace[..k], \edge)$ be true iff the edge $\edge$ of the specification automaton was derived from a trigger for which the monitor reports a violation for the trace $\trace$ up to the $k$th step.
Lastly, $\mode_{i,k}$ are fresh modes. 
\begin{alignat*}{2}
	\modes_{k+1} &= \modes_k \cup \bigcup_{i}\Set{\mode_{i,k}}
	&\actions_{k+1} = \actions_k \cup \bigcup_{i}\Set{\action^{\trace_i}_k} \\
  \edges_{k+1} &= \edges_k \cup \bigcup_i\Set{(\modemap_k(\trace_i), \action^{\trace_i}_k, \mode_{i,k})} 
    &\modemap_{k+1}(\trace_i) = \mode_{i,k}
\end{alignat*}
\begin{align*}    
  \guard_{k+1}(\edge) &= \twopartdef%
  { \guard_k(\edge) }%
  { \edge \in \edges_k }%
  { \guardspec(\alpha_{k+1}(\edge)) }\\
  \flow_{k+1}(\mode) &= 
    \twopartdef%
    { \solve(x^{\trace_i}_k, x^{\trace_i}_{k+1}, \delay^{\trace_i}_k) }%
    { \mode = \mode_{i,k} }
    { \flow_{k}(\mode) }
\end{align*}
\[\alpha_{k+1}(\mode) = \threepartdef%
  { \mode^\alpha }%
  { \exists \trace\colon \spec(\trace[..k], (\alpha_k(\modemap_k(\trace)), \action^{\trace}_k, \mode^\alpha)) }
  { \alpha_k(\mode) }
  { \mode \in \modes_k }
  { \alpha_k(\modemap_k(\trace_i)) }[ \mode = \mode_{i,k} ]\]
  
\end{definition}

Intuitively, for each position of each trace the construction 
\begin{enumerate*}[(i)]
\item adds a new mode with the dynamics exhibited by the delay transition, 
\item adds a new edge from $\mode$ to $\mode'$ for the discrete transition, and 
\item updates the mode maps accordingly.
\end{enumerate*}
The latter means that if the transition was accompanied by a step in $\finaut^\spec$, $\alpha$ maps the $\mode'$ to the respective abstract mode and looks up the guard from the specification.
Otherwise, it maps $\mode'$ to the same abstract state as $\mode$ with a vacuous guard indicating a lack of information.

\subsection{Merging Modes}
Evidently, following the procedure yields an automaton with at most $\card\trace \cdot \card\traces$ modes arranged as a tree as can be seen in \Cref{fig:overview}. 
It transformed the overly coarse specification automaton into an overly fine hybrid automaton.
To find the sweet spot between both extremes, the next construction step merges modes within an abstract state provided they are sufficiently similar.
Suppose some relation~$\actionsim$ captures this notion of similarity.
Then, intuitively, the construction deems any two modes $\mode \not\actionsim \mode'$ sufficiently \emph{dis}similar such that they must represent different modes in the original system.
For this, let $\refinerel$ denote the relation induced by $\alpha$ for a constructed hybrid automaton, $\mode_1 \refinerel \mode_2$ indicates that both modes refine the same abstract state, \ie, \(\alpha(\mode_1) = \alpha(\mode_2)\).

\begin{definition}[Action Similarity]\label{def:merges:action}
For a constructed hybrid automaton $\constraut$, two modes $\mode_1, \mode_2 \in \modes^+$ are action-similar if they share some discrete characteristics and reside in the same abstract state of the specification. 
Assume there are some modes $\mode_1', \mode_2' \in \modes^+$ and action $\action \in \actions^+$.a
\begin{align*}
  \mode_1 \actionsim \mode_2 &\iff 
  \alpha(\mode_1) = \alpha(\mode_2) \land {} \\
  &\quad\quad(\Set{(\mode_1', \action, \mode_1),(\mode_2', \action, \mode_2)} \subseteq \edges^+ \lor{}\\
  &\quad\quad\Set{(\mode_1, \action, \mode_1'), (\mode_2, \action, \mode_2')} \subseteq \edges^+)  
\end{align*}
\end{definition}
Note that by construction $\refinerel$ is coarser than $\actionsim$. 

Terminal modes need further attention:  consider the automaton in \Cref{fig:ha:drone}.
There are two identical traces in the language of the automaton starting in \textsc{Takeoff} and traversing \textsc{Straight}, but ending in different \textsc{Landing} modes.
Based on these traces the construction cannot distinguish the two terminal modes, since the difference in modes is unobservable.    
In fact, there is no finite set of traces for which they can be distinguished with certainty.
This forces the construction to merge them as can be seen in \Cref{fig:ha:drone_output}.

\begin{definition}[Terminal and Merge Similarity]\label{def:merges:terminal}
  Two terminal modes are \emph{terminal-similar} iff they reside in the same abstract state.
\[  \mode_1 \termsim \mode_2 \iff \outdeg(\mode_1) = \outdeg(\mode_2) = 0 \land \alpha(\mode_1) = \alpha(\mode_2) \]
Two modes are \emph{merge-similar} iff they are either action-similar or terminal-similar: $ \mergesim\ =\ \actionsim \cup \termsim $
\end{definition}

\paragraph{Merge Operation.}  
The merge operation now minimizes the automaton with respect to~$\mergesim$ by building the quotient automaton.
Formally, a merge operates on an equivalence relation $\bisim$ over the set of modes where each equivalence class $\eqclass \subseteq \modes$ will be replaced by a single representative $\representative[\bisim]\eqclass$. 
By slight abuse of notation let $\representative[\bisim]\mode = \representative[\bisim]\eqclass$ for $\mode \in \eqclass$.
Moreover, if context permits, the subscript may be omitted.
The representative conserves the language of each mode contained in $\eqclass$ by retaining discrete transitions and computing the convex hull for its continuous components. 

\begin{definition}[Merge Automaton]\label{def:merges:commit}
  Merging an automaton $\aut$ with respect to an equivalence relation $\bisim$ yields an automaton $\commitpartition{\aut}{\bisim}$ where all elements of an equivalence class get merged into a single element retaining its language.

\begin{alignat*}{2}
  &\commitpartition{\modes}{\bisim} = \Set{\representative\mode \given \mode \in \modes} 
    &\commitpartition{\initialstate}{\bisim} = (\representative{\initial\mode}, \initial{x}) \\
  &\commitpartition{\flow}{\bisim}(\representative\eqclass) = \Convex{\bigcup_{\mode \in \eqclass}\Set{\flow(\mode)}}
  	&\commitpartition{\modemap}{\bisim}(\trace) = \representative{\modemap(\trace)}\\
  &\commitpartition{\edges}{\bisim} = \Set*{(\representative{\mode_1}, \action, \representative{\mode_2}) \given (\mode_1,\action,\mode_2) \in \edges} \quad
  	&\commitpartition{\alpha}{\bisim}(\representative\mode) = \alpha(\mode) \\
  &\commitpartition{\actions}{\bisim} = \Set{\lambda \given \exists \edge \in \commitpartition{\edges}{\bisim}\colon \edge = (\representative{\eqclass_1}, \action, \representative{\eqclass_2})}
\end{alignat*}
\vspace{-1.4\baselineskip} 
\begin{equation*}
  \commitpartition{\guard}{\bisim}((\eqclass_1, \action, \eqclass_2)) = \Convex{\Set{(\mode_1, \action, \mode_2) \given \mode_1 \in \eqclass_1 \land \mode_2 \in \eqclass_2}}\qquad\thinspace
\end{equation*}
\end{definition}

\subsection{Putting it Together: Construction Algorithm}
The overall construction algorithm now proceeds as outlined in \Cref{alg:overallconstruction}: 
First, the procedure extracts information from the specification, constructs the initial automaton and refines it successively by iterating over the traces. 
After processing all traces completely, the procedure computes and applies the merges with respect to action- and terminal-similarity.

\begin{algorithm}[t]
  \begin{algorithmic}[1]
    \Require Specification \spec, Traces \traces
    \State Extract $\finaut^\spec, \guardspec$ from $\spec$ \Comment{\Cref{def:constr:spec}}
    \State Construct $\constraut_1, \modemap_1, \alpha_1$ from \traces and $\guardspec$ \Comment{\Cref{def:constr:init}}
    \For{$k$ from $1$ to $\card\trace$ for $\trace \in \traces$}
      \State Update to $\constraut_k, \modemap_k, \alpha_k$ \Comment{\Cref{def:constr:step}}
    \EndFor
    \State Compute the action-similarity $\actionsim$ \Comment{\Cref{def:merges:action}}
    \State Compute the merge-similarity $\mergesim$ \Comment{\Cref{def:merges:terminal}}
    \State Compute $\constraut = \commitpartition{\constraut_{\card\traces}}{(\mergesim)}$ \Comment{\Cref{def:merges:commit}}
\end{algorithmic}
\caption{Construct Conservative Hybrid Automaton}
\label{alg:overallconstruction}
\end{algorithm}

\paragraph{Time Complexity.}\label{sec:complexity}

The construction process consists of three phases: extraction, construction and merging.
Recall that the dimensionality, \ie, the number of continuous state variables is $n$.
The first phase scales linearly in the size of the specification \(\bigO{\card\spec}\).
The second phase construct an automaton with a single mode per step of any trace.
Its size and the running time of the construction scales linearly with the number and length of traces as it creates a new mode per step of any trace.
It is also linear in the dimension since the dynamics of each dimension have to be computed separately per mode.
Hence, the complexity is in \(\bigO{n\cdot\card\trace\cdot\card\traces}\).
Lastly, the complexity of the last phase depends on the complexity of a single merge, which is linear in the dimension, and the number of merges.
The latter is quadratic in the size of $\constraut_{\card\traces}$, which in turn is linear in the number and length of traces: \(\bigO{n\cdot\card*{\constraut_{\card\traces}}^2} = \bigO{n\cdot\card\trace^2\cdot\card\traces^2}\).
This, however, only describes the worst case. 
The procedure compares each mode against each other with respect to $\mergesim$.
In the best case, all elements of an equivalence class are identified successively by chance.
In this case, the process is quadratic in the number of equivalence classes: \(\Omega(n\cdot\card\mergesim^2)\).
Here, \(\card\mergesim\) denotes the number of equivalence classes induced by $\mergesim$ with \(\card*{\finaut^\spec} \leq \card*\mergesim = \card*{\modes^+} \leq \card*{M^+_{\card\traces}} \).
In conclusion, the overall asymptotic running time is dominated by the merge procedure: \(\bigO{n\cdot\card\trace^2\cdot\card\traces^2}\)

%% file: source/completeness.tex
\section{Correctness of Construction}\label{sec:correctness}
The validation of the construction requires a proof that the constructed automaton is --- under certain assumptions on the input --- indeed conservative.
For this, a key criterion is that the automaton over-approximates the discrete and continuous behavior of the original system when projected down to the parts that contributed to the inputs.
Evidently, if the original system encompasses parts that were neither reflected in the specification, nor traversed in the input traces, the constructed automaton cannot reconstruct it.

Hence, this section first formalizes requirements on the input data.
Then, a definition of projection automata enables proving that the constructed automaton subsumes the language of the projected original system.

\subsection{Requirements on Input Data}\label{sec:inputreqs}
The construction of the conservative automaton relies on the quality of the input traces and specification.
Hence, they need to satisfy three criteria:
\begin{enumerate*}[(i)] 
  \item the specification must be an abstraction of the real system, 
  \item its trigger conditions must be at least as restrictive as the respective conditions on mode changes, and 
  \item the trace set needs to traverse every control mode of the system sufficiently often to capture the discrete behavior.
\end{enumerate*}  

\begin{definition}[Adequacy of Input Data]\label{def:adequacy}
A specification \spec and trace set \traces are \emph{adequate} for a hybrid automaton \aut iff they satisfy three criteria.
\begin{enumerate}[ref=\thedefinition.\arabic*]
  \item\label{def:adequacy:partition} The specification induces a coarser automaton $\finaut^\spec$ than the original, \ie, $\exists \specrel\colon \commitpartition{\aut}{\specrel} \discretebisim \finaut^\spec$.
  \item\label{def:adequacy:guard} 
  Assume $V^\spec = \Set{\representative{\eqclass} \given \eqclass \in \partition}$. 
  For any discrete transition that is both in \aut and $\finaut^\spec$, the specification contains a mode change condition that is at least as permissive as the guard of the respective transition in \aut.
  Formally, let $\mode^\spec_1 \not\specrel \mode^\spec_2$ and $\mode^\spec_1 \discretebisim \mode_1 \land \mode^\spec_2 \discretebisim \mode_2$.
  Then, for all $(\mode^\spec_1, \action, \mode^\spec_2) \in \edges^\spec$ and $\edge = (\representative[\specrel]{\mode_1}, \action, \representative[\specrel]{\mode_2}) \in \commitpartition{\edges}{\specrel}$:
	\[
		\guard(\mode, \action, \mode) \implies \guardspec(e)
	\]
    \item\label{def:adequacy:traceset} For every mode $\mode$ in $\aut$, let $a_{\mode} = \indeg(\mode) + \outdeg(\mode)$ be the number of input and output actions of $\mode$ in $\traces$. 
    The trace set needs to contain more than $a_\mode(a_\mode-1)/2$ traversals through $\mode$.
    Here, a trace \trace traverses through a mode if its omniscient counterpart $\godtrace$ contains two subsequent edges first ending and then starting from~$\mode$.
    \begin{align*}
    	\forall \mode \in \modes_{\aut}: \card*{&\Set*{(\edge^\godtrace_i,\edge^\godtrace_{i+1})\mid \exists \godtrace\in\godtraces, i < \card\godtrace, \action, \action' \in \actions_{\aut}\colon  \\&(\edge^\godtrace_i,\edge^\godtrace_{i+1}) = ((\mode_1, \action, \mode),(\mode, \action', \mode_2)) 
        }} > \frac{(a_{\mode})(a_{\mode}-1)}{2}
    \end{align*}
  \end{enumerate}
\end{definition}

Evidently, these criteria depend on the original hybrid automaton, which seemingly contradicts the premise of this paper regarding the unavailability of such a model.
Nevertheless, the criteria are designed in a way that they are either satisfied naturally or can be satisfied without access to all formal details of the system.

First, consider \ref{def:adequacy:partition} and \ref{def:adequacy:guard}.
These criteria restrict the specification, which was hand-crafted for the underlying system. 
Here, a reasonable specification summarizes control modes that are subject to the same requirements;  at the same time, the specification needs to capture changes in the abstract state precisely to impose the correct sub-specification on the system.
Thus, even without perfect knowledge of the inner workings and dynamics of the system, the first two criteria can be ensured.
Consider the third criterion, which is concerned with the trace set.
A thorough testing process demands that all discrete paths\footnote{This is the case for terminating systems.  In the presence of infinite discrete paths, a threshold on the length of executions is usually imposed.} through the system are tested at least once.  
Moreover, the system has a fixed control interface, represented by $\actions$.
As a result, it is reasonable to assume that the number of times each control mode is traversed during the development exceeds the threshold required by Criterion~\ref{def:adequacy:traceset}.
This again does not rely on knowledge about the exact mode structure nor dynamics of the underlying system.

The exact threshold for the third criterion seems arbitrary but is anchored in graph theory, the impact of which can be seen in the next lemma.
\begin{lemma}[Trace Connectivity]\label{lem:connectivity}
  Let~$\spec$ and $\traces$ be adequate for $\aut$. For any mode~$\mode$ in~$\aut$ with incoming edge label $\action_i$ and outgoing edge label $\action_o$, there is a mode $\mode'$ in~$\constraut$ with the very same edge labels and $\alpha(\mu) = \alpha(\mu')$.
\end{lemma}
\begin{proof}
By reduction on the graph connectivity problem. 
Let $\mathcal{G}(\mode, \traces) = (V, E)$ be a graph where $V$ is the set of labels of incoming or outgoing edges of $\mode$ in \aut. 
$E$ connects two action labels $\action,\action'$ if there is a trace in $\traces$ that reaches $\mode$ via $\action$ and leaves via $\action'$.
The solution of the graph connectivity problem states that $\action$ and $\action'$ are necessarily connected if $\card E$ exceeds $\card V \cdot (\card V - 1) / 2$.
This threshold corresponds to Criterion \ref{def:adequacy:traceset}.
Recall that $\actionsim$ relates all modes with at least one common incoming or outgoing edge label.
Thus, since $\actionsim \subseteq \mergesim$, all respective modes are merged in $\constraut$ and by \Cref{def:merges:commit}, the resulting mode $\representative[\mergesim]{\mode}$ retains these transitions.
Lastly, $\mergesim$ refines $\specrel$, hence $\alpha(\mode) = \alpha(\representative[\mergesim]\mode)$, which concludes the proof.
\end{proof}

\subsection{Projection Automata}

The assessment of the quality of the reconstruction depends on the projection of the original system onto the set of traces.
This first requires a definition of projections on automata.

\begin{definition}[Projection Automata]\label{def:proj:aut}
The projection of an automaton $\aut$ down to a set of omniscient traces $\godtraces$ is an automaton $\projection{\aut}$ with the following constituents.
\begin{gather*}
  \projection{\modes} = \bigcup_{\godtrace \in \godtraces} \bigcup_{i \leq \card{\godtrace}} \Set{\mode^\godtrace_i \given \mode^\godtrace_i \in \modes}\\
  \projection{\actions} = \bigcup_{\godtrace \in \godtraces} \bigcup_{i < \card{\godtrace}} \Set{\action^\godtrace_i \given \action^\godtrace_i \in \actions } \\
  \projection{\edges} = \Set{(\mode, \action, \mode') \in \edges \given \mode, \mode' \in \projection{\modes} \land \action \in \projection{\actions}}\\
  \projection{\guard(\edge)} = [\nu^{\min}_{\edge}, \nu^{\max}_{\edge}] \quad\quad
  \projection{\flow(\mode)} = [\nu^{\min}_{\mode}, \nu^{\max}_{\mode}] \quad\quad   \projection{\initialstate} = \initialstate 
\shortintertext{
Here, for $\phi \in \Set{\min,\max}$, the min and max values for guards and flows are:} 
\begin{align*}
    \nu^\phi_{\edge} &= \phi\Set*{ x \given \exists \godtrace, \exists i < \card\godtrace\colon x = x^\godtrace_i \land \edge = \edge^\godtrace_i}\\
	 \nu^\phi_{\mode} &= \phi\Set*{ f \given \exists \godtrace, \exists i < \card\godtrace\colon \edge^\godtrace_i = (\mode', \action, \mode) \land x^\godtrace_i + \delay^\godtrace_i f = x^\godtrace_{i+1}} 
\end{align*}
\end{gather*}
\end{definition}

Intuitively, the projection strips the automaton of any information not reflected in the set of traces.
This reduces the sets of modes, edges, and transition labels.  
By definition, the initial state occurs in all traces and thus remains the same. 
Guards and flows are reduced to the maximum and minimum value exhibited by some trace.

Note that the projection automaton $\projection{\aut}$ is not meant to be constructed at any point;  it merely serves as theoretical point of reference for the quality of the construction. 
It is easy to see that in general the projection reduces the expressiveness of an automaton, \ie, $\lang{\aut} \supseteq \lang{\projection{\aut}}$.
This, however, is not necessarily the case as the following theorem shows.

\begin{restatable}[Perfect Projection]{theorem}{perfectprojection}\label{thm:perfect:projection}
For any hybrid automaton~\aut there is a finite set of traces for which the projection onto these traces yields the identity, \ie,  $\exists \godtraces^\ast \subseteq \lang{\aut}\colon \traces^\ast \text{ finite}\land \lang{\projection{\aut}[\godtraces^\ast]} = \lang\aut$.
\end{restatable}

The proof of \Cref{thm:perfect:projection} can be found in \Cref{sec:appendix:correctness}.
Note that the language equality cannot be extended to identical or isomorphic automata since \aut can contain unreachable modes that are not reflected in its language and thus not in any trace.
The theorem emphasizes the generality of the conservative construction:
For an appropriate trace set, the projection of an automaton perfectly resembles the original system.
Since the constructed automaton is conservative with respect to this very projection, it is also conservative with respect to the original system. 
This is independent of the exact structure of the underlying system.

\subsection{Construction Guarantees}

The first observations are that application of a merge and iterations of the construction do not reduce the language of an automaton.

\begin{restatable}[Lossless Merge]{lemma}{losslessmerge}
\label{lem:lossless:merge}
Given a constructed hybrid automaton $\constraut$ and an equivalence relation $\bisim$, merging $\constraut$ with respect to $\bisim$ yields a more permissive automaton, \ie, $\lang{\constraut} \subseteq \lang{\commitpartition{\constraut}{\bisim}}$.
\end{restatable}

The proof of the lossless merge can be found in \Cref{sec:appendix:correctness}.
After showing the property of a merged automaton, we state that our construction is lossless.
\begin{lemma}[Lossless Construction]\label{lem:lossless:construction}
Given a set of traces \traces and specification~\spec. 
For any iterations $i$ and $j$, if $j \geq i$, then the set of edges, the flow, and the transition guards only grow over the iterations:
\[ \edges_i \subseteq \edges_j \land \flow_i \subseteq \flow_j \land \forall \edge \in \edges_i\colon \guard_i(\edge) \subseteq \guard_j(\edge) \]
\end{lemma}
\begin{proof}
This lemma follows directly from the construction step (\Cref{def:constr:step}).	
\end{proof}

This suffices to prove that the language of the constructed automaton at least includes all input traces.

\begin{restatable}[Input Trace Inclusion]{theorem}{inputtraceinclusion}
\label{thm:input:trace:inclusion}
Given an adequate set of traces \traces and specification \spec, the language of a constructed automaton $\constraut$ subsumes \traces, \ie, $\traces \subseteq \lang{\constraut}$.
\end{restatable}

The proof of \Cref{thm:input:trace:inclusion} can be found in \Cref{sec:appendix:correctness}.
A stronger classification of the language of $\constraut$ requires some insight into its discrete structure in relation to the projection automaton of the original system.
Specifically, $\projection{\aut}$ has a finer discrete structure than $\constraut$.

\begin{restatable}[Discrete Refinement]{lemma}{discreterefinement}
\label{lem:discrete:refinement}
For an adequate specification~\spec and set of traces~\traces for a hybrid automaton \aut, the reconstruction $\constraut$ is coarser than the projection of~\aut onto~\traces, \ie, $\exists \reconrel\colon \commitpartition{\constraut}{\reconrel} \discretebisim \projection{\aut}$.
\end{restatable}

  \begin{proof}
  The proof proceeds in two steps.
  First, for an arbitrary trace $\trace$ through $\projection\aut$ it generates a trace $\trace'$ through $\constraut$.
  Second, it constructs the equivalence relation $\reconrel$ based on these trace pairs.

  \textsc{Step 1:}
  For a given $\godtrace \in \lang{\projection\aut}$, the proof inductively constructs $\godtrace'$ with $\godtrace' \in \lang{\constraut}$ such that the observable traces for $\godtrace$ and $\godtrace'$ are equal.
  Moreover, for any step $i$: $\modemap(\godtrace[0..i]) = \modemap(\godtrace'[0..i])$  .
The induction base is trivial since both traces originate in the fixed initial state, which corresponds to the initial state of the specification automaton.
\textsc{Induction Step:}
  Suppose the observable traces for $\godtrace[0..i]$ and $\godtrace'[0..i]$ are equal and $\modemap(\godtrace[0..i]) = \modemap(\godtrace'[0..i])$.  
  Suppose further the last action label was $\action$ and the next is $\action'$.
  
  Since the label combination $\action$, $\action'$ appears in $\godtrace$, it is also present in a mode $\mode$ in $\constraut$ by \Cref{lem:connectivity} with $\alpha(\godtrace[0..i]) = \alpha(\godtrace'[0..i]) = \alpha(\mode'')$ .
  Further, $\modemap(\godtrace[0..i])$ has an incoming $\action$ label, so by \Cref{def:merges:action}, $\modemap(\godtrace[0..i]) \actionsim \mode''$ due to their shared action label.
  Hence, by \Cref{lem:lossless:merge}, $\alpha(\godtrace[0..1])$ has an outgoing edge with label $\action'$, which proves that the discrete edge is present.
  
  Now, it suffices to show that $\modemap(\godtrace[0..i]) = \modemap(\godtrace'[0..i])$.
  There are possibilities:  \emph{both} traces take a transition in terms of $\modemap$ or neither one does.  
  This follows from Criterion \ref{def:adequacy:guard} and \Cref{lem:lossless:construction} stating that both automata are refinements of $\finaut^\spec$.
  If both take such a transition, the claim follows because $\finaut^\spec$ is deterministic and both automata refine $\finaut^\spec$.
  If neither one does, both remain in the same state in $\finaut^\spec$; the claim follows from the induction hypothesis.
  \textsc{Step 2:}
  The trace pairs $\godtrace, \godtrace'$ induce an equivalence relation:
  \[ \reconrel\ = \Set{(\modemap(\godtrace[0..i]), \modemap(\godtrace'[0..i])) \given i \in \natural} \]
  This relation is a witness for the claim that $\projection\aut$ is a refinement of $\constraut$ due to the trace inclusion proven in Step 1.  

  \end{proof}	

  Note that the refinement can be a true refinement since the merge criterion might falsely relate modes that are distinct in $\projection\aut$ but share some discrete behavior, as discussed before.

\begin{corollary}\label{cor:mergesim:refines:proj}
  $\reconrel$ is finer than $\mergesim$.
\end{corollary}

\begin{theorem}[Conservative Construction]\label{thm:conservative:construction}
Let $\aut$ be a hybrid automaton with an adequate set of traces \traces and specification \spec.
The constructed automaton $\constraut$ over-approximates the language of the projection automaton $\projection\aut$: $\lang{\projection{\aut}} \subseteq \lang{\constraut}$.
\end{theorem}

\ifthenelse{\boolean{hideproofs}}{}{
  \begin{proof}
  Let $\godtrace \in \lang{\projection{\aut}}$.  
  The proof constructs a trace $\trace \in \lang\constraut$ such that $\trace$ is an observable counterpart for $\godtrace$.
  The initial real-valued state of $\trace$ is $x_0^\godtrace$ with $\modemap(\trace[1]) = \initial{\mode^+}$.
  By construction, $\initial{\mode^+} \reconrel \projection{\initial\mode} = \mode_0^\godtrace$.
  
  For the induction, consider a delay transition $x^\godtrace_i, \mode^\godtrace_i, \delay^\godtrace_i, x^\godtrace_{i+1}$ where $\mode^\godtrace_i \reconrel \mode^+$ by \Cref{lem:discrete:refinement}.  
  Let $x^\godtrace_i + f \delay^\godtrace_i = x^\godtrace_{i+1}$ for $f \in \rn$.
  By definition of the projection automaton (\Cref{def:proj:aut}), \traces contains traces $\trace^\uparrow$ and $\trace^\downarrow$ exhibiting the flow $f^\uparrow$ and $f^\downarrow$ at point $a^\uparrow \leq \card*{\trace^\uparrow}$ and $a^\downarrow \leq \card*{\trace^\downarrow}$, respectively, while traversing $\mode^\godtrace_i$ with $f^\downarrow \leq f \leq f^\uparrow$.
   By construction, there are modes $\mode^{a^\uparrow} \in \modes^+_{a^\uparrow}$ and $\mode^{a^\downarrow} \in \modes^+_{a^\downarrow}$ with $flow_{a^\uparrow}(\mode^{a^\uparrow}) = f^\uparrow$ and $flow_{a^\downarrow}(\mode^{a^\downarrow}) = f^\downarrow$.
  \Cref{lem:lossless:construction,lem:lossless:merge} guarantee that further construction steps and merges retain this information. 
  Moreover, \Cref{lem:discrete:refinement} implies that $\mode^{a^\uparrow}$, $\mode^{a^\downarrow}$, and $\mode^+$ are equal with respect to $\reconrel$.
  By \Cref{cor:mergesim:refines:proj}, they are also equal with respect to $\mergesim$.
  Thus, $\flow(\mode^+) = \Convex{\Set{\flow(\mode)\given\mode\in\eqclass^+}}$ with $f^\downarrow, f^\uparrow, f \in \flow(\mode^+)$.
  As a result, $\trace$ may contain the subsequence $x^\godtrace_i, \delay^\godtrace_i, x^\godtrace_{i+1}$ representing the delay transition.
  
  For discrete transitions, consider $\edge = (\mode^\godtrace_i, \action^\godtrace_i, \mode^\godtrace_{i+1})$. 
  We show that $\edge^+ = (\mode_s^+, \action^\godtrace_i, \mode_t^+)$ is a valid transition assuming that $\mode_s^+ = \modemap(\trace[0..i])$, \ie, the trace constructed so far ended in $\mode_s^+$.
  There are three cases:
  
  \textsc{Case a)} Both $\mode^\godtrace_i \specrel \mode^\godtrace_{i+1}$ and $\mode_s^+ \refinerel \mode_t^+$. 
  Intuitively, this means that both automata remain in the same state of the specification automaton.
  In this case, by construction: $\guard_i(\edge^+) = \guardspec(\alpha_{i}(\edge^+)) = \neutralelem$ and by \Cref{lem:lossless:construction,lem:lossless:merge}: $\guard_i(\edge^+) \subseteq \guard^+(\edge^+)$.  
  Thus, the guard is satisfied trivially.
  The existence of the edge in the constructed automaton follows from \Cref{lem:discrete:refinement}.
  
  \textsc{Case b)} Neither $\mode^\godtrace_i \specrel \mode^\godtrace_{i+1}$ nor $\mode_s^+ \refinerel \mode_t^+$. 
  Intuitively, this means neither automaton remains in the same state of the specification automaton.
  In this case, $\guard_i(\edge^+) = \guardspec(\alpha_{i}(\edge^+))$.
  By Criterion \ref{def:adequacy:guard} and \Cref{def:proj:aut}, we know that $\projection{\guard}(\edge) \implies \guard(\edge)$ and $\guard(\edge) \implies  \guardspec(\alpha_{i}(\edge^+))$.
  Again, by \Cref{lem:lossless:construction} and \Cref{lem:lossless:merge} we know $\guard_i(\edge^+) \subseteq \guard^+(\edge^+)$ and the existence of the edge in the constructed automaton follows from \Cref{lem:discrete:refinement}.
  
  \textsc{Case c)} Either $\mode^\godtrace_i \not\specrel \mode^\godtrace_{i+1}$ or $\mode_s^+ \not\refinerel \mode_t^+$ but not both.
  This case is impossible for adequate specifications (Definition \ref{def:adequacy:partition}) and by the definition of $\specrel/\refinerel$.
  
  Thus, the discrete transition exists and is applicable in the reconstructed automaton.
  This concludes the proof.
  \end{proof}
}

%% file: source/evaluation.tex
\begin{figure*}[t]
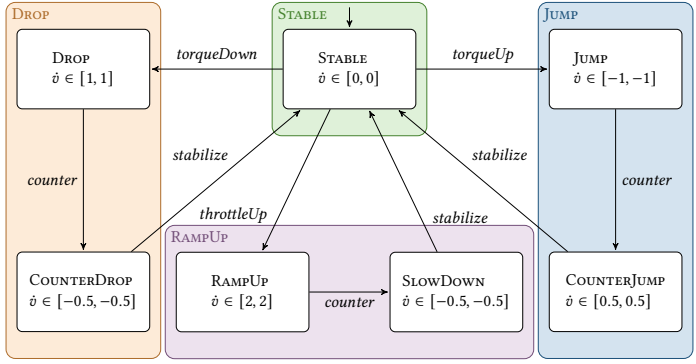
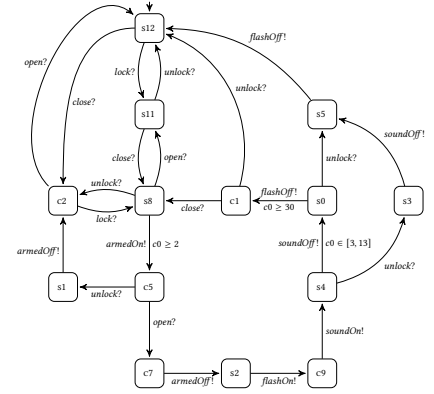

	\begin{subfigure}[b]{0.52\linewidth}
	\centering
    	\include{source/figures/borzoo_bumps}
    	\caption{Hybrid automaton approximating the engine timing control system. A single trace of length ten enables perfect reconstruction. Colored states indicate the specification automaton.}
    	\label{fig:borzoo:bumps}
	\end{subfigure}
	\quad\quad\quad
	\begin{subfigure}[b]{0.33\linewidth}
	\centering
  	   \include{source/figures/cas}
  	   \caption{Timed automaton for the car alarm system. Perfect reconstruction requires seven traces of length twelve.}
  	   	\label{fig:cas}
    \end{subfigure}
	\caption{
		Example automata by Medhat~et~al.~\cite{MiningFramework}~(a) and Tappler~et~al.~\cite{learningtimedgenetic}~(b).
	}	
	\label{fig:automata}
\end{figure*}

\section{Experiments}\label{sec:eval}
The empirical evaluation shows the scalability and precision of the approach presented in this paper. 
It is based on a prototype implementation in Rust\footnote{\url{https://www.rust-lang.org}} and the code is open source.
All experiments were conducted on an Intel i5-7200u with 8GB RAM.

\subsection{Aircraft System}\label{sec:eval:aircraft}

As a first proof of concept, consider the running example from \Cref{fig:ha:drone}.
For adequate input traces, the output will always be structurally equal with varying dynamics.
This can be seen in \Cref{fig:ha:drone_output}, which shows the results of two construction processes.
The dynamics in black are constructed from three hand-picked traces of length eight.
Two of these traces travers all three \textsc{Travel} modes, whereas the last one skips the course adjustment modes and loops in one of the \textsc{Landing} modes instead.
 As can be seen, by picking the state values for the traces in such a way that they represent the extreme behavior, the reconstruction of the dynamics is perfect.
 Conversely, the constructed dynamics based on an adequate trace set of ten traces obtained by conducting random walks on the original system is shown in gray.
The traces can be found in \Cref{app:aircraft:traces}.
Evidently, the reconstruction closely resembles the original system both structurally and in terms of dynamics despite being based on a small set of random traces.

\subsection{Scalability}\label{sec:eval:scalability}

Recall the complexity of each step of the construction algorithm, \ie, extraction, construction, and merging, from \Cref{sec:complexity}.
The extraction only requires a single pass over the specification and is thus negligible.
The construction and merges depend on the dimensionality of the system and the number and length of traces. 
The merge also depends on the number of equivalence classes with respect to $\mergesim$ in the best case, which is the size of the output automaton.

For this reason, the scalability evaluation considers exactly these three factors: dimensionality, number and length of traces, and output-size.
To this end, it automatically generates an automaton with matching specification and adequate trace set.
The automaton is shaped like a binary tree of variable depth $d$ (scales the length of traces) where each of the $2^{d+1}-1$ nodes is a control mode with dynamics of variable dimension (scales the dimensionality).
The specification summarizes a variable number of modes with equal depth (scales the output size) and generates a variable number of adequate traces (scales the number of traces) enabling the respective merges.

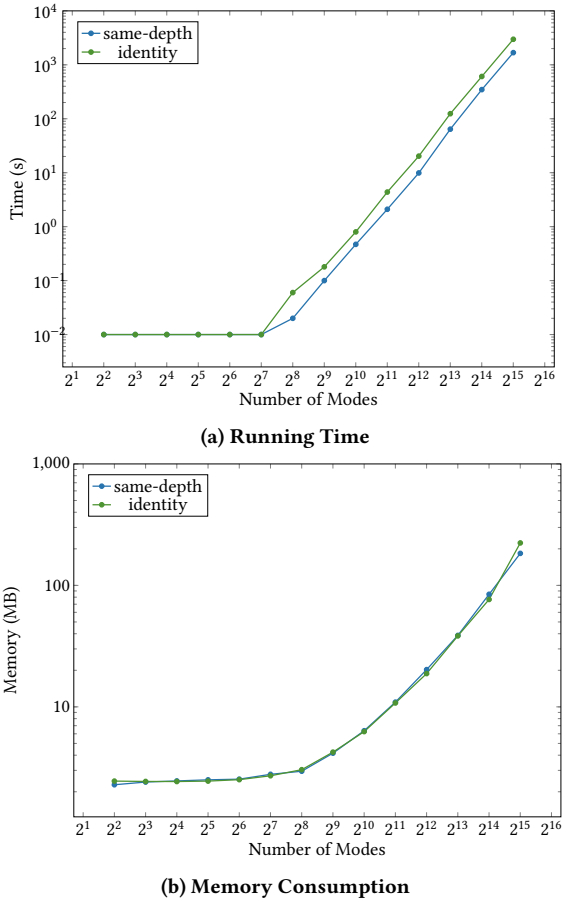
\begin{figure}[t]
	\begin{subfigure}[b]{0.49\textwidth}
		\centering
		\input{source/figures/depth_time}
		\caption{Running Time
		}
		\label{fig:eval:scalability:depth_time}
	\end{subfigure}\hfill
	\begin{subfigure}[b]{0.49\textwidth}
		\centering
		\input{source/figures/depth_mem}
		\caption{Memory Consumption
		}
		\label{fig:eval:scalability:depth_mem}
	\end{subfigure}
	\caption{The results of the scalability analysis for different sizes of the original automaton and specifications enabling many (blue) or no (green) merges.
	}
	\label{fig:eval:scalability:all}
\end{figure}

\Cref{fig:eval:scalability:depth_time} and \Cref{fig:eval:scalability:depth_mem} show the running time and memory consumption for varying sizes of the original automaton.
The green line represents runs where $\mergesim$ only equates identities, prohibiting any merges.
For the blue line, two modes are equal if they have the same depth in the underlying automaton.
The number and size of the traces required for an adequate trace set scales linearly with the size and the depth, respectively.
Independent of the existence of merges, the running time lies below a second for automata with less than $2^{10}$ modes and increases steadily afterwards --- as expected considering the asymptotic complexity.
Even for automata with $2^{15}$ modes, the construction terminated after less than an hour (green line) or half an hour (blue line).
The memory consumption behaves similarly, starting to rise significantly around $2^7$ owing both to the increased number of traces stored in memory, and the resulting size of $\constraut_{\card\traces}$.
Note that the memory consumption almost exclusively stems from the construction process; merging only deallocates memory.
In terms of running time, the lion's share comes from merging due to the difference in time-complexity. 
At a size of $2^9$ modes the running time of the construction process amounts to 3.061\%  and decreases to 0.015\% for automata with $2^{15}$ modes.

The dimensionality impacts the running time to a lesser extent.
Raising the dimension from 1,000 to 7,000 for an automaton size of $10^{10}$ increases the running time from around \SI{4}{\second} to \SI{14}{\second}. 
The limiting factor here is the memory consumption: each additional dimension increases the memory consumption of every guard condition, mode, and trace step.	
As a result, 10,000 dimensions requires around \SI{5}{\giga\byte} of memory, which is also reflected in a relatively steep increase in running time to \SI{96}{\second}.
Lastly the number of traces has an almost identical impact on both the running time and memory consumption.
The scale of the impact lies in-between the one of the dimensionality and the size of the automaton.
Raising the number of traces from 3,000 to 16,000 increases the running time roughly 40-fold.
Detailed results for the dimensionality and number of traces can be found in \Cref{app:eval}.

\subsection{Comparison Against Other Approaches}

The construction of a hybrid automaton requires the determination of both the discrete structure and continuous behavior.
Regarding the former, the conservative construction relies on the information provided by the specification and refines abstract states according to trace information.
This restricts revisions to a local level. 
In absence of such a specification, other approaches resort to learning algorithms.
Medhat~\etal~\cite{MiningFramework} use a modification of Angluin's $L^\ast$ algorithm~\cite{Angluin} to learn the discrete structure separately from the dynamics.
Tappler~\etal~\cite{learningtimedgenetic} on the other hand learn a timed automaton with genetic programming.
Note that these automata alleviate the need to infer dynamics. 
We will evaluate the conservative construction against both of these approaches.

\paragraph{Anguin's $L^\ast$ and Clustering.}
Medhat~\etal~\cite{MiningFramework} use an adaptation of $L^\ast$ and clustering to identify the discrete and continuous behavior of a system, respectively.
In their case study, they use a Simulink model of a closed-loop engine timing control system.
\Cref{fig:borzoo:bumps} shows an approximate representation of the system as hybrid model.
Their construction uses eight traces to generate an automaton that resembles the underlying system for new traces up to an error of 2.6\%.
Note that a detailed comparison is exacerbated by a lack of information regarding the running time, memory consumption, and length of the input traces.
The conservative construction can perfectly reconstruct the automaton with one hand-picked trace of length ten within less than \SI{1}{\milli\second}.
When using random walks, an average of 35 traces of length fifteen suffices for the perfect reconstruction.
In this example, the specification always summarizes modes belonging to a certain operation of the system, \ie, a drop, jump, ramp-up and the stable configuration.

\paragraph{Genetic Programming.}
Tappler~\etal~\cite{learningtimedgenetic} use genetic programming to successively adapt a candidate automaton to encompass all input traces.
As an example, they consider a timed automaton modeling a car alarm system~(see \Cref{fig:cas}).
A sufficiently precise reconstruction requires 2,000 randomly generated traces and took a mean of around \SI{100}{\minute}.
When using seven hand-selected traces, the conservative construction can perfectly reconstruct the system within less than \SI{1}{\milli\second}, disregarding resets.
With random walks, an average of 35 traces of length fifteen is necessary for the perfect reconstruction.
Note that the nature of the car alarm system does not allow for a meaningful summary of several modes in the specification.
However, in terms of performance, the specification is irrelevant owing to the small size of the system.

%% file: source/figures/borzoo_bumps.tex
\tikzstyle{state}=[draw, rectangle, fill=none, align=left, minimum width=2.5cm, 
minimum height=1.5cm, rounded corners=1mm]
\resizebox{1\linewidth}{!}{
  \begin{tikzpicture}[->,>=stealth',shorten >= 1pt,auto,node distance=4.2cm,every node/.style={transform shape}]
  
  \node[state, fill=white, initial above, initial text=] (stablespeed) {%
  	\modename{Stable} \\
  	$\dot{v} \in [0,0]$
  };
  
  \node[state, fill=white] (jump) [right=2.5 of stablespeed] {%
  	\modename{Jump} \\
  	$\dot{v} \in [-1,-1]$
  };
  
  \node[state, fill=white] (drop) [left=2.5 of stablespeed] {%
  	\modename{Drop} \\
  	$\dot{v} \in [1,1]$
  };
  @
  \node[state, fill=white] (rampup) [below left=2.69 and -0.5 of stablespeed] {%
  	\modename{RampUp} \\
  	$\dot{v} \in [2,2]$
  };
  
  \node[state, fill=white] (counterdrop) [below of=drop] {%
  	\modename{CounterDrop} \\
  	$\dot{v} \in [-0.5,-0.5]$
  };
  
  \node[state, fill=white] (counterjump) [below of=jump] {%
  	\modename{CounterJump} \\
  	$\dot{v} \in [0.5,0.5]$
  };
  
  \node[state, fill=white] (slowdown) [below right=2.69 and -0.5 of stablespeed] {%
  	\modename{SlowDown} \\
  	$\dot{v} \in [-0.5,-0.5]$
  };
  
  \begin{scope}[-,on background layer]
    \filldraw[fill=prefix!20, draw=prefix!70!black, rounded corners=5pt]
      ($(stablespeed.north west)+(-0.2,0.5)$)  -- ($(stablespeed.north east)+(0.2,0.5)$) -- ($(stablespeed.south east)+(0.2,-0.5)$) -- ($(stablespeed.south west)+(-0.2,-0.5)$) -- cycle node[anchor=north west] at ($(stablespeed.north west)+(-0.2,0.5)$) {$\color{prefix!70!black} \textsc{Stable}$};
  \end{scope}
  
  \begin{scope}[-,on background layer]
    \filldraw[fill=prelude!20, draw=prelude!70!black, rounded corners=5pt]
      ($(drop.north west)+(-0.2,0.5)$) -- ($(drop.north east)+(0.2,0.5)$) -- ($(counterdrop.south east)+(0.2,-0.5)$) -- ($(counterdrop.south west)+(-0.2,-0.5)$) -- cycle node[anchor=north west] at ($(drop.north west)+(-0.2,0.5)$) {\color{prelude!70!black} \textsc{Drop}};
  \end{scope}
  
  \begin{scope}[-,on background layer]
    \filldraw[fill=loop!20, draw=loop!70!black, rounded corners=5pt]
      ($(jump.north west)+(-0.2,0.5)$) -- ($(jump.north east)+(0.2,0.5)$) -- ($(counterjump.south east)+(0.2,-0.5)$) -- ($(counterjump.south west)+(-0.2,-0.5)$) -- cycle node[anchor=north west] at ($(jump.north west)+(-0.2,0.5)$) {\color{loop!70!black} \textsc{Jump}};
  \end{scope}
  
  \begin{scope}[-,on background layer]
    \filldraw[fill=postfix!20, draw=postfix!70!black, rounded corners=5pt]
      ($(rampup.north west)+(-0.2,0.5)$) -- ($(slowdown.north east)+(0.2,0.5)$) -- ($(slowdown.south east)+(0.2,-0.5)$) -- ($(rampup.south west)+(-0.2,-0.5)$) -- cycle node[anchor=north west] at ($(rampup.north west)+(-0.2,0.5)$) {\color{postfix!70!black} \textsc{RampUp}};
  \end{scope}
  
  \path (stablespeed) edge node[above] {%
  		$\mathit{torqueUp}$
  	  } (jump)
  	  (stablespeed) edge node[above] {%
  	  	$\mathit{torqueDown}$
  	  } (drop)
  	  (stablespeed) edge node[pos=.84,above left] {%
  	  	$\mathit{throttleUp}$
  	  } (rampup)
  	  (drop) edge node[left] {%
  	  	$\mathit{counter}$
  	  } (counterdrop)
  	  (counterdrop) edge node[pos=.6,above left] {%
  	  	$\mathit{stabilize}$
  	  } (stablespeed)
  	  (jump) edge node[right] {%
  	  	$\mathit{counter}$
  	  } (counterjump)
  	  (counterjump) edge node[pos=.6,above right] {%
  	  	$\mathit{stabilize}$
  	  } (stablespeed)
  	  (rampup) edge node[below] {%
  	  	$\mathit{counter}$
  	  } (slowdown)
  	  (slowdown) edge node[pos=.16,above right] {%
  	  	$\mathit{stabilize}$
  	  } (stablespeed);
  
  \end{tikzpicture}
}

%% file: source/figures/cas.tex
\tikzstyle{state}=[draw, rectangle, fill=none, align=left, minimum width=1cm, 
minimum height=1cm, rounded corners=1mm]
\resizebox{1\linewidth}{!}{
\begin{tikzpicture}[->,>=stealth',shorten >= 1pt,auto,node distance=3cm,every node/.style={transform shape},scale=0.5]

\node[state, initial above, initial text=] (s12) {%
	\large \modename{s12}
};

\node[state] (s11) [below of=s12] {%
	\large \modename{s11}
};

\node[state] (s8) [below of=s11] {%
	\large \modename{s8}
};

\node[state] (c2) [left of=s8] {%
	\large \modename{c2}
};

\node[state] (c5) [below of=s8] {%
	\large \modename{c5}
};

\node[state] (s1) [left of=c5] {%
	\large \modename{s1}
};

\node[state] (c7) [below of=c5] {%
	\large \modename{c7}
};

\node[state] (s2) [right of=c7] {%
	\large \modename{s2}
};

\node[state] (c9) [right of=s2] {%
	\large \modename{c9}
};

\node[state] (s4) [above of=c9] {%
	\large \modename{s4}
};

\node[state] (s0) [above of=s4] {%
	\large \modename{s0}
};

\node[state] (s3) [right of=s0] {%
	\large \modename{s3}
};

\node[state] (s5) [above of=s0] {%
	\large \modename{s5}
};

\node[state] (c1) [left of=s0] {%
	\large \modename{c1}
};

\path (s12) edge[bend right=20] node[left] {%
	  	\large $\mathit{lock?}$
	  } (s11)
	  (s11) edge[bend right=20] node[right] {%
	  	\large $\mathit{unlock?}$
	  } (s12)
	  (s11) edge[bend right=20] node[left] {%
	  	\large $\mathit{close?}$
	  } (s8)
	  (s8) edge[bend right=20] node[right] {%
	  	\large $\mathit{open?}$
	  } (s11)
	  (s8) edge[bend right=20] node[above] {%
	  	\large $\mathit{unlock?}$
	  } (c2)
	  (c2) edge[bend right=20] node[below] {%
	  	\large $\mathit{lock?}$
	  } (s8)
	  (s12) edge[out=180,in=90] node[pos=.6,right] {%
	  	\large $\mathit{close?}$
	  } (c2)
	  (c2) edge[out=135,in=140] node[left] {%
	  	\large $\mathit{open?}$
	  } (s12)
	  (s8) edge node[left] {%
	  	\large $\mathit{armedOn!}$
	  } node[right] {%
	  	\large $c0 \geq 2$
	  } (c5)
	  (c5) edge node[below] {%
	  	\large $\mathit{unlock?}$
	  } (s1)
	  (s1) edge node[below left] {%
	  	\large $\mathit{armedOff!}$
	  } (c2)
	  (c5) edge node[above right] {%
	  	\large $\mathit{open?}$
	  } (c7)
	  (c7) edge node[below] {%
	  	\large $\mathit{armedOff!}$
	  } (s2)
	  (s2) edge node[below] {%
	  	\large $\mathit{flashOn!}$
	  } (c9)
	  (c9) edge node[right] {%
	  	\large $\mathit{soundOn!}$
	  } (s4)
	  (s4) edge[bend right=30] node[below right] {%
	  	\large $\mathit{unlock?}$
	  } (s3)
	  (s3) edge[bend right=30] node[above right] {%
	  	\large $\mathit{soundOff!}$
	  } (s5)
	  (s4) edge node[left] {%
	  	\large $\mathit{soundOff!}$
	  } node[right] {%
	  	\large $c0 \in [3,13]$
	  } (s0)
	  (s0) edge node[right] {%
	  	\large $\mathit{unlock?}$
	  } (s5)
	  (s5) edge[bend right=25] node[above right] {%
	  	\large $\mathit{flashOff!}$
	  } (s12)
	  (s0) edge node[above] {%
	  	\large \large $\mathit{flashOff!}$
	  } node[below] {%
	  	\large $c0 \geq 30$
	  } (c1)
	  (c1) edge node[below] {%
	  	\large $\mathit{close?}$
	  } (s8)
	  (c1) edge[out=75,in=340] node[above right] {%
	  	\large $\mathit{unlock?}$
	  } (s12);

\end{tikzpicture}
}

%% file: source/figures/depth_time.tex
\tikzstyle{layerpoint}=[circle,draw=loop,fill=loop,inner sep=0pt,minimum size=0.15cm]
\tikzstyle{idpoint}=[circle,draw=prefix!80!black,fill=prefix!80!black,inner sep=0pt,minimum size=0.15cm]
\pgfplotsset{compat=1.5}
\centering
\resizebox{0.85\linewidth}{!}{
	\begin{tikzpicture}[font=\Huge]
		\begin{axis}[
			ymax=10000,
			height=1.5\textwidth,
			width=2\textwidth,
			xlabel=Number of Modes,
			xmode=log,
			log basis x={2},
			ylabel=Time (\SI{}\second),
			yticklabel style={text width=width("$10^{-100}$"),align=left},
			ymode=log,
			log basis y={10},
			legend pos=north west
		]
		\addplot[mark=*,color=loop,very thick] coordinates {
			(4,0.01)
			(8,0.01)
			(16,0.01)
			(32,0.01)
			(64,0.01)
			(128,0.01)
			(256,0.02)
			(512,0.1)
			(1024,0.47)
			(2048,2.1)
			(4096,9.9)
			(8192,64.14)
			(16384,347.15)
			(32768,1681.66)
		};
		\addplot[mark=*,color=prefix!80!black,very thick] coordinates {
			(4,0.01)
			(8,0.01)
			(16,0.01)
			(32,0.01)
			(64,0.01)
			(128,0.01)
			(256,0.06)
			(512,0.18)
			(1024,0.8)
			(2048,4.38)
			(4096,20.3)
			(8192,124.13)
			(16384,608.11)
			(32768,2975.65)
		};
		\legend{same-depth, identity}
			
		\end{axis}
	\end{tikzpicture}
}

%% file: source/figures/depth_mem.tex
\tikzstyle{layerpoint}=[circle,draw=loop,fill=loop,inner sep=0pt,minimum size=0.15cm]
\tikzstyle{idpoint}=[circle,draw=prefix!80!black,fill=prefix!80!black,inner sep=0pt,minimum size=0.15cm]
\pgfplotsset{%
	compat=1.5,
	log y ticks with fixed point/.style={
      yticklabel={
        \pgfkeys{/pgf/fpu=true}
        \pgfmathparse{10^\tick}%
        \pgfmathprintnumber[fixed relative, precision=3]{\pgfmathresult}
        \pgfkeys{/pgf/fpu=false}
      }
  }
}
\centering
\resizebox{0.87\linewidth}{!}{
	\begin{tikzpicture}[font=\Huge]
		\begin{axis}[
			ymax=1000,
			height=1.5\textwidth,
			width=2\textwidth,
			xlabel=Number of Modes,
			xmode=log,
			log basis x={2},
			log y ticks with fixed point,
			ylabel=Memory (\SI{}{\mega\byte}),
			ymode=log,
			log basis y={10},
			legend pos=north west
		]
		\addplot[mark=*,color=loop,very thick] coordinates {
			(4,2.284)
			(8,2.404)
			(16,2.456)
			(32,2.508)
			(64,2.54)
			(128,2.784)
			(256,2.952)
			(512,4.156)
			(1024,6.364)
			(2048,10.94)
			(4096,20.264)
			(8192,38.772)
			(16384,84.496)
			(32768,183.464)
		};
		\addplot[mark=*,color=prefix!80!black,very thick] coordinates {
			(4,2.448)
			(8,2.432)
			(16,2.428)
			(32,2.448)
			(64,2.516)
			(128,2.708)
			(256,3.04)
			(512,4.24)
			(1024,6.256)
			(2048,10.764)
			(4096,18.792)
			(8192,38.408)
			(16384,76.576)
			(32768,223.62)
		};
		\legend{same-depth, identity}
			
		\end{axis}
	\end{tikzpicture}
}

%% file: source/related_work.tex
\section{Related Work}\label{sec:rw}
The theory of hybrid automata was first studied by Henzinger~\cite{TheoryOfHybridAutomata} as the real-time extension of timed automata~\cite{TheoryOfTimedAutomata}. 
Learning the complex structure of timed and hybrid automata is a line of research that resulted in deterministic and stochastic reconstruction algorithms.

Niggemann~\etal~\cite{LearningBehaviourModels} present the tool \textsc{HyBUTLA} that builds prefix trees of the traces and applies merges when appropriate.
Since Angluin's $L^\ast$ algorithm~\cite{Angluin} is a prominent solution for learning discrete automata, several extensions for timed automata were proposed~\cite{learningoneclock,event-recording-automata}.
Based on that, Medhat~\etal~\cite{MiningFramework} split the learning process of a hybrid automaton into two steps.
They first learn the discrete model of the automaton with $L^\ast$ and then capture the dynamics using clustering.
Both of these techniques can potentially be replaced or integrated into different frameworks.
Hence, their approach is complementary to the conservative construction: substituting the clustering for the simpler $\solve$ function can yield better precision at the price of conservativeness.
Focusing on the medical application domain, \textsc{HyMN}~\cite{HyMN} learns patient specific parameters for hybrid automata deterministically.
Soto~\etal~\cite{synthesis} synthesize a hybrid automaton with an online algorithm without relying on a specification as discrete template.
While precision is very high and completeness is shown, learning a trace prompts a global analysis of the previously learned hybrid automaton, which incurs performance penalties.
The conservative construction avoid this complexity by using the specification automaton and the adequacy criterion.
This way, revisions are local and still retain correctness guarantees.
Other approaches for learning hybrid automata using mathematical models for node identification were proposed by Summerville~\etal~\cite{CHARDA} and Breschi~\etal~\cite{linearseperation}.

If large datasets of traces are available, stochastic learning of hybrid automata is feasible.
Tappler~\etal~\cite{learningtimedgenetic} use genetic programming to reconstruct timed automata both in an offline and online setting~\cite{learningtimedgenetic2}.
Santana~\etal~\cite{PHA} build hybrid automata with the Expectation-Maximization algorithm to iteratively define the model parameters.
An unsupervised learning approach was presented by Lee~\etal~\cite{UnsupervisedLearning}, whereas Birgelen and Niggemann~\cite{MLHybridSystems} use self-organizing feature maps.
Despite the success of machine learning, the results do not provide provable guarantees.

%% file: source/conclusion.tex
\section{Conclusion}\label{sec:conclusion}

This paper presented a construction algorithm for conservative hybrid automata from development artifacts in the shape of a runtime monitoring specification and pre-recorded execution traces.
The construction is validated mathematically by proving that the result is an over-approximation under certain assumptions on the inputs. 
An additional empirical evaluation revealed both the extraordinary scalability of the construction and that even randomly generated inputs regularly satisfy the input requirements.
Considering that --- in a realistic setting --- these inputs are high-quality artifacts acquired during development of the system, they should no longer be left under-utilized.
Treating them as the valuable assets they are allows for constructing precise, conservative hybrid automata in a scalable fashion.
%

%

%% file: source/appendix.tex
\begin{appendix}

\section{Adequate Traces of Aircraft System}\label{app:aircraft:traces}

This section of the appendix presents the adequate traces used to learn the version of the aircraft system depicted in \Cref{fig:ha:drone_output} with the dynamics written in black.

\newcommand\scalefactortraces{.7}
\noindent
\begin{minipage}[t]{0.33\linewidth}
  \scalebox{\scalefactortraces}{
    \begin{tabular}{ll}
      & (0,0,0) \\
      $\myrightarrow{\mathit{cruise},300}$ & (300,0,300) \\
      $\myrightarrow{\mathit{turnL},5}$ & (750,0,290) \\
      $\myrightarrow{\mathit{LtoS},5}$ & (1200,-750,280) \\
      $\myrightarrow{\mathit{turnL},5}$ & (1650,-750,270) \\
      $\myrightarrow{\mathit{LtoS},5}$ & (2100,-1500,260) \\
      $\myrightarrow{\mathit{turnR},5}$ & (2550,-1500,250) \\
      $\myrightarrow{\mathit{RtoS},5}$ & (3000,-1500,240) \\
      $\myrightarrow{\mathit{descend},5}$ & (3450,-1500,230) \\
      $\myrightarrow{5}$ & (3450,-1500,150)
    \end{tabular}
  }
\end{minipage}
\begin{minipage}[t]{0.33\linewidth}
  \scalebox{\scalefactortraces}{
    \begin{tabular}{ll}
      & (0,0,0) \\
      $\myrightarrow{\mathit{cruise},10}$ & (1000,0,300) \\
      $\myrightarrow{\mathit{turnR},5}$ & (2500,0,310) \\
      $\myrightarrow{\mathit{RtoS},5}$ & (4000,750,320) \\
      $\myrightarrow{\mathit{turnR},5}$ & (5500,750,330) \\
      $\myrightarrow{\mathit{RtoS},5}$ & (7000,1500,340) \\
      $\myrightarrow{\mathit{turnL},5}$ & (8500,1500,350) \\
      $\myrightarrow{\mathit{LtoS},5}$ & (10000,1500,360) \\
      $\myrightarrow{\mathit{descend},5}$ & (11500,1500,370) \\
      $\myrightarrow{5}$ & (12500,1500,370)
    \end{tabular}
  }
\end{minipage}
\begin{minipage}[t]{0.33\linewidth}
  \scalebox{\scalefactortraces}{
    \begin{tabular}{ll}
      & (0,0,0) \\
      $\myrightarrow{\mathit{cruise},20}$ & (1000,0,300) \\
      $\myrightarrow{\mathit{descend},5}$ & (2000,0,300) \\
      $\myrightarrow{\mathit{adjust},5}$ & (2375,0,275) \\
      $\myrightarrow{\mathit{adjust},5}$ & (2750,0,250) \\
      $\myrightarrow{\mathit{adjust},5}$ & (3125,0,225) \\
      $\myrightarrow{\mathit{adjust},5}$ & (3500,0,200) \\
      $\myrightarrow{\mathit{adjust},5}$ & (3875,0,175) \\
      $\myrightarrow{\mathit{adjust},5}$ & (4250,0,150) \\
      $\myrightarrow{5}$ & (4625,0,125)
    \end{tabular}
  }
\end{minipage}

\section{Proofs of Section 5}\label{sec:appendix:correctness}
This section contains the missing proofs of \Cref{sec:correctness}.

\perfectprojection*
\begin{proof}
  The proof selects traces from $\lang\aut$ enabling the perfect projection.
  For each $\edge$ in $\edges$, $\godtraces^\ast$ contains a trace $\godtrace$ with $\edge \in \godtrace$.
  This immediately entails that the projected edge set, set of actions, set of modes, and initial mode are accurate when disregarding unreachable entities.
  For each mode, $\godtraces^\ast$ encompasses four traces per dimension: one minimizing and one maximizing the flow and continuous state variable of the mode and dimension.  
  The minimization and maximization is over the set of traces rather than over the mode itself.  
  As a result, the projection of the flow is perfect.
  Lastly, for each mode, outgoing edge, and dimension, there are two traces in $\godtraces^\ast$ which maximize and minimize the continuous state value before taking the transition.
  Hence, the projection of the guard condition is lossless in terms the language of the automaton.

  In conclusion, $\lang{\projection{\aut}[\godtraces^\ast]} = \lang\aut$ with \[\card*{\godtraces^\ast} = \card \edges + n(5\card \modes + \sum_{\mode \in \modes} \outdeg(\mode))\] where $n$ is the dimension of $\aut$.  \end{proof}
  
\losslessmerge*

\begin{proof} By contradiction: 	Assume there is a trace $\trace \in \lang{\constraut} \setminus \lang{\commitpartition{\constraut}{\bisim}}$.
  As the merge operation is defined by unifying modes, $\trace$ either (1) takes a discrete transition or (2) traverses a continuous state not permitted in the merged automaton.
  \Cref{def:merges:commit} \enquote{bends} edges such that they originate and end in the respective representatives.  
  Hence, the construction retains all edges up to elimination of duplicates due to set semantics, ruling out (1).
  Regarding (2), \Cref{def:merges:commit} builds the convex hull for all flow and guard definitions of the merged states. 
  The convex hull is at least as permissive as its constituents, rendering a less permissive behavior impossible, which concludes the proof. 
\end{proof}

\inputtraceinclusion*

  \begin{proof}
  Let $\trace \in \traces$ be an arbitrary input trace.
  We first prove by induction that any subsequence of length $i$ of $\trace$ is included in the language of $\constraut_{i}$.  
  Recall that $\modemap_i(\trace)$ denotes the mode in which $\trace[0,i]$ ends.
  
  \textsc{Induction Base: $i = 0$}.  
  By construction $(x^\trace_0, \initial\mode) \in \lang{\constraut_{0}}$.
  Moreover, $\modemap_0(\trace) = \initial\mode$ marks the (start and) end point of the trace.
  
  \textsc{Induction Step: $i-1 \rightarrow i$}.  Consider $\trace[0,i] = \trace[0,i-1],\action_i,\delay_i,x_i$.
  By induction hypothesis, $\trace[0,i-1] \in \lang{\constraut_{i-1}}$.
  By \Cref{lem:lossless:construction}, the language membership carries over to $\constraut_i$.
  By construction, $\edge = (\mode_{i-1}, \action_i, \modemap_i(\trace)) \in \edges_i$ with guard 
  $\guard_i(\trace) = \guardspec(\edge^\alpha)$. 
  Here, $\guardspec(\edge^\alpha)$ is either $\top$ if the transition is not present \spec, or the condition from the respective trigger in \spec. 
  In the latter case, by construction of $\alpha_i$, the respective trigger (and thus guard) condition is satisfied in $x_{i-1}$.
  This enables the discrete transition and ensuring that the trace ends in $\modemap_i(\trace)$.
  The delay transition is valid because of the definition of $\solve$.
  Hence, $\trace \subseteq \lang{\constraut_{\card\trace}}$.
  By \Cref{lem:lossless:merge}, this result carries over to $\constraut$, \ie, $\trace \subseteq \lang{\constraut}$
 \end{proof}

\section{Additional Evaluations Results}\label{app:eval}
This section presents supplementary material for the evaluation in \Cref{sec:eval}.

\Cref{fig:suppl:dim} shows that the construction scales extremely well for an increasing dimension, both in terms of running time and memory consumption assuming the number of modes is a constant $2^{11}-1$.
The increase in running time mainly stems from the merge operation, which computes the convex hull of the dynamics for each dimension separately.
However, the impact is rather low since the running time is less than \SI{100}{\second} even for $10^4$ dimensions.
Similarly, the memory consumption increases solely owing to the necessity to store the multi-dimensional dynamics of each mode.

\Cref{fig:suppl:trace} shows the result when constructing five-dimensional automata with the same number of modes, but varying number of traces.
Each trace has a length of eleven.
Increasing the number of traces results in an increase in running time since the construction has to traverse the current automaton for each traces.
It also increases the memory consumption because traces lead to the creation of new modes.
However, the construction scales well in both performance metrics with a running time of around \SI{17}{\minute} and memory consumption of \SI{200}{\mega\byte} for 16,000 traces.

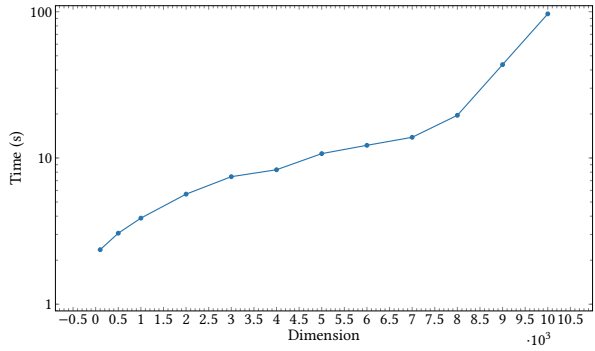
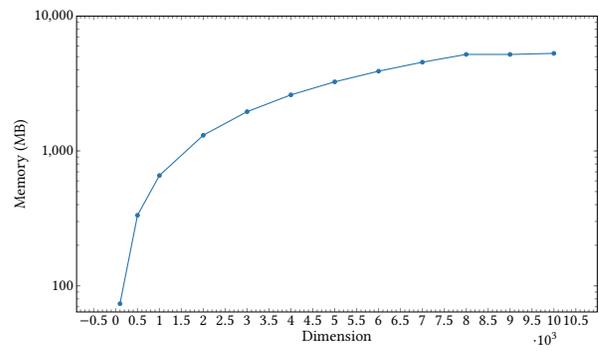
\begin{figure*}
	\begin{subfigure}[b]{.49\linewidth}
		\centering
		\input{source/figures/dimension_time}
		\caption{Running Time
		}
		\label{fig:eval:scalability:dimension_time}
	\end{subfigure}\hfill
	\begin{subfigure}[b]{.49\linewidth}
		\centering
		\input{source/figures/dimension_mem}
		\caption{Memory Consumption
		}
		\label{fig:eval:scalability:dimension_mem}
	\end{subfigure}
	\caption{Performance of the reconstruction for varying dimensions. 
	The original automaton has $2^{11}-1$ modes and each of the around 1,000 trace has length eleven.
	}
	\label{fig:suppl:dim}
\end{figure*}

\begin{figure*}
	\begin{subfigure}[b]{.49\linewidth}
		\centering
		\input{source/figures/traces_time}
		\caption{Running Time
		}
		\label{fig:eval:scalability:traces_time}
	\end{subfigure}\hfill
	\begin{subfigure}[b]{.49\linewidth}
		\centering
		\input{source/figures/traces_mem}
		\caption{Memory Consumption
		}
		\label{fig:eval:scalability:traces_mem}
	\end{subfigure}
	\caption{Performance of the reconstruction for a varying number of traces. 
	The original automaton has $2^{11}-1$ modes and five dimensions. Each trace has length eleven.
	}
	\label{fig:suppl:trace}
\end{figure*}
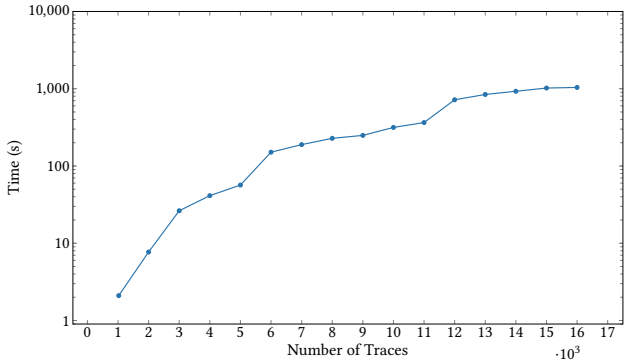
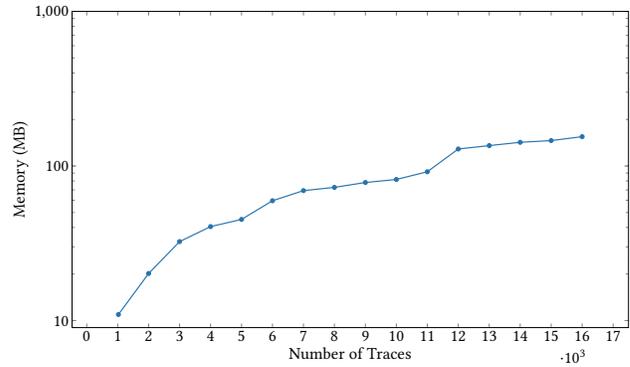
\end{appendix}

%% file: source/figures/dimension_time.tex
\tikzstyle{layerpoint}=[circle,draw=loop,fill=loop,inner sep=0pt,minimum size=0.15cm]
\tikzstyle{idpoint}=[circle,draw=prefix!80!black,fill=prefix!80!black,inner sep=0pt,minimum size=0.15cm]
\pgfplotsset{compat=1.5}
\centering
\resizebox{.9\linewidth}{!}{
	\begin{tikzpicture}[font=\Huge]
		\begin{axis}[
			ymin=0.9,
			ymax=110,
			height=1.5\textwidth,
			width=2.5\textwidth,
			xlabel=Dimension,
			scaled x ticks=base 10:-3,
			minor x tick num=4,
			log ticks with fixed point,
			ylabel=Time (\SI{}\second),
			ymode=log,
			log basis y={10},
		]
		\addplot[mark=*,color=loop,very thick] coordinates {
			(100,2.36)
			(500,3.06)
			(1000,3.88)
			(2000,5.66)
			(3000,7.45)
			(4000,8.31)
			(5000,10.71)
			(6000,12.21)
			(7000,13.85)
			(8000,19.59)
			(9000,43.51)
			(10000,96.73)
		};
			
		\end{axis}




	\end{tikzpicture}
}

%% file: source/figures/dimension_mem.tex
\tikzstyle{layerpoint}=[circle,draw=loop,fill=loop,inner sep=0pt,minimum size=0.15cm]
\tikzstyle{idpoint}=[circle,draw=prefix!80!black,fill=prefix!80!black,inner sep=0pt,minimum size=0.15cm]
\pgfplotsset{compat=1.5}
\centering
\resizebox{.9\linewidth}{!}{
	\begin{tikzpicture}[font=\Huge]
		\begin{axis}[
			ymin=64,
			ymax=10000,
			height=1.5\textwidth,
			width=2.5\textwidth,
			xlabel=Dimension,
			scaled x ticks=base 10:-3,
			minor x tick num=4,
			log ticks with fixed point,
			ylabel=Memory (\SI{}{\mega\byte}),
			ymode=log,
			log basis y={10},
		]
		\addplot[mark=*,color=loop,very thick] coordinates {
			(100,73.68)
			(500,333.512)
			(1000,657.736)
			(2000,1306.864)
			(3000,1955.964)
			(4000,2604.82)
			(5000,3253.928)
			(6000,3903.516)
			(7000,4551.976)
			(8000,5198.092)
			(9000,5194.988)
			(10000,5290.536)
		};
			
		\end{axis}




	\end{tikzpicture}
}

%% file: source/figures/traces_time.tex
\tikzstyle{layerpoint}=[circle,draw=loop,fill=loop,inner sep=0pt,minimum size=0.15cm]
\tikzstyle{idpoint}=[circle,draw=prefix!80!black,fill=prefix!80!black,inner sep=0pt,minimum size=0.15cm]
\pgfplotsset{compat=1.5}
\centering
\resizebox{0.95\linewidth}{!}{
	\begin{tikzpicture}[font=\Huge]
		\begin{axis}[
			ymax=10000,
			height=1.5\textwidth,
			width=2.5\textwidth,
			xlabel=Number of Traces,
			scaled x ticks=base 10:-3,
			log ticks with fixed point,
			ylabel=Time (\SI{}\second),
			ymode=log,
			log basis y={10},
		]
		\addplot[mark=*,color=loop,very thick] coordinates {
			(1024,2.1)
			(2000,7.7)
			(3000,26.35)
			(4000,41.31)
			(5000,56.65)
			(6000,150.97)
			(7000,189.08)
			(8000,228.05)
			(9000,248.92)
			(10000,315.56)
			(11000,364.89)
			(12000,719.37)
			(13000,843.03)
			(14000,927.37)
			(15000,1020.91)
			(16000,1039.21)
		};
			
		\end{axis}




	\end{tikzpicture}
}

%% file: source/figures/traces_mem.tex
\tikzstyle{layerpoint}=[circle,draw=loop,fill=loop,inner sep=0pt,minimum size=0.15cm]
\tikzstyle{idpoint}=[circle,draw=prefix!80!black,fill=prefix!80!black,inner sep=0pt,minimum size=0.15cm]
\pgfplotsset{compat=1.5}
\centering
\resizebox{0.95\linewidth}{!}{
	\begin{tikzpicture}[font=\Huge]
		\begin{axis}[
			ymin=9,
			ymax=1000,
			height=1.5\textwidth,
			width=2.5\textwidth,
			xlabel=Number of Traces,
			scaled x ticks=base 10:-3,
			log ticks with fixed point,
			ylabel=Memory (\SI{}{\mega\byte}),
			ymode=log,
			log basis y={10},
		]
		\addplot[mark=*,color=loop,very thick] coordinates {
			(1024,10.94)
			(2000,20.16)
			(3000,32.412)
			(4000,40.536)
			(5000,45.116)
			(6000,59.576)
			(7000,69.28)
			(8000,72.688)
			(9000,78.188)
			(10000,81.78)
			(11000,91.752)
			(12000,128.748)
			(13000,135.46)
			(14000,142.248)
			(15000,145.956)
			(16000,154.712)
		};
			
		\end{axis}

	\end{tikzpicture}
}